\documentclass{extarticle}
\usepackage{geometry}
\usepackage[hyphens]{url}  %

\usepackage{graphicx} %
\usepackage{natbib}  %
\usepackage{caption} %
\usepackage{algorithm}
\usepackage{algpseudocode}

\usepackage{newfloat}
\usepackage{listings}
\DeclareCaptionStyle{ruled}{labelfont=normalfont,labelsep=colon,strut=off} %
\floatstyle{ruled}
\newfloat{listing}{tb}{lst}{}
\floatname{listing}{Listing}

\usepackage{amssymb}
\usepackage{nicefrac}
\usepackage{amsthm}
\usepackage{amsmath}
\usepackage{mathtools}
\usepackage{thm-restate}

\usepackage[T1]{fontenc}
\usepackage{fontawesome5}

\usepackage{tabularx}
\usepackage{booktabs}
\usepackage{multirow}

\usepackage[inline]{enumitem}

\usepackage{xspace}
\usepackage{xparse}

\usepackage[dvipsnames]{xcolor}
\usepackage{hyperref}
\hypersetup{
    colorlinks   = true,
    linkcolor    = red, %
    urlcolor     = teal, %
    citecolor    = blue %
}
\usepackage{cleveref}
\crefname{claim}{claim}{claims}
\Crefname{claim}{Claim}{Claims}

\usepackage{etoolbox} %
\AtBeginEnvironment{appendices}{%
  \crefname{section}{appendix}{appendices}%
  \Crefname{section}{Appendix}{Appendices}%
}

\usepackage{newtxtext,newtxmath}
\definecolor{cbBlue}{RGB}{102,153,204}
\definecolor{cbOrange}{RGB}{230,159,0}
\definecolor{cbRed}{RGB}{153,68,85}

\usepackage{tikz}
\usetikzlibrary{calc,backgrounds,shapes,decorations,patterns,decorations.pathreplacing,positioning}
\tikzset{
 reload/.style = {
  rectangle,
  rounded corners,
  draw=none,
  fill=black!75!white,
  text=white,
  inner sep=2pt
 }
}

\newcommand{\agent}{a}

\newcommand{\threePart}{\probName{3-Partition}}
\newcommand{\maxThreeSatThree}{\probName{Max 3 SAT-3}}
\newcommand{\optBasicProb}{\probName{Routing Plan$^*$}}

\newcommand{\LRed}{\textsc{L}-reduction\xspace}

\newcommand{\threePartDummy}{y}

\newcommand{\N}{\mathbb{N}}
\newcommand{\naturals}{\N}

\newcommand{\rest}{\!\!\restriction}

\newcommand{\Yes}{\textsf{Yes}\xspace}
\newcommand{\No}{\textsf{No}\xspace}

\newcommand{\sV}{s}
\newcommand{\tV}{t}
\newcommand{\numVertices}{n}

\newcommand{\agents}{A}
\newcommand{\numAgents}{m}

\newcommand{\traps}{R}
\newcommand{\numTraps}{\tau}
\newcommand{\reloadFn}{c}

\newcommand{\destroyed}{\dagger}

\newcommand{\plan}{\ensuremath{\rho}}

\newcommand{\goal}{\ensuremath{k}}

\DeclareMathOperator{\OperLen}{len}
\newcommand{\len}{\ensuremath{\OperLen}}

\DeclareMathOperator{\OperSnap}{snap}
\newcommand{\snap}{\ensuremath{\OperSnap}}

\newcommand{\probName}[1]{\textsc{#1}\xspace}
\newcommand{\basicProb}{\probName{Routing Plan}}

\newcommand{\Oh}[1]{\ensuremath{\mathcal{O}\!\left(#1\right)}}

\NewDocumentCommand{\cc}{ O{} O{} m }{\mbox{%
    \expandafter\ifx\expandafter\relax\detokenize{#2}\relax\else{#2-}\fi%
    \textsf{#3}%
    \expandafter\ifx\expandafter\relax\detokenize{#1}\relax\else{-#1}\fi%
    }\xspace}

\newcommand{\NP}{\cc{NP}}
\newcommand{\NPh}{\cc[hard]{NP}}
\newcommand{\NPhness}{\cc[hardness]{NP}}
\newcommand{\NPc}{\cc[complete]{NP}}
\newcommand{\PSc}{\cc[complete]{PSPACE}}

\newcommand{\APX}{\cc{APX}}
\newcommand{\APXh}{\cc[hard]{APX}}
\newcommand{\APXhness}{\cc[hardness]{APX}}

\newtheorem{theorem}{Theorem}
\newtheorem{observation}{Observation}
\newtheorem{lemma}{Lemma}
\newtheorem{definition}{Definition}

\newtheorem{example}{Example}
\newtheorem{corollary}{Corollary}
\newtheorem{claim}{Claim}
\newenvironment{claimproof}[1]{\noindent\emph{Proof.}\hspace{0.15cm}#1}{\hfill$\blacktriangleleft$\smallskip}
\newtheorem{remark}{Remark}

\newcommand{\defProblemQuestion}[3]{\medskip%
  \noindent\begin{tabularx}{\linewidth}{lX}
	\toprule
	\multicolumn{2}{c}{#1} \\\midrule
	\small\emph{Input:} & \small{} #2 \\
	\small\emph{Question:\hspace{-0.25cm}} & \small{} #3 \\
	\bottomrule
\end{tabularx}
}

\title{Optimal Path Planning in Hostile Environments}

\author {
    Andrzej Kaczmarczyk,\textsuperscript{\rm 1}
    Šimon Schierreich,\textsuperscript{\rm 2}\\
    Nicholas Axel Tanujaya,\textsuperscript{\rm 3}
    Haifeng Xu\textsuperscript{\rm 4}\vspace{0.25cm}\\
    \begin{tabular}{p{0.35\textwidth}p{0.25\textwidth}}
      \small \textsuperscript{\rm 1} Czech Technical University in Prague &
      \small \textsuperscript{\rm 3} Bina Nusantara University \\
      \small \textsuperscript{\rm 2} AGH University of Krakow &
      \small \textsuperscript{\rm 4} University of Chicago \\
    \end{tabular}\vspace{0.25cm}\\
    {\small\texttt{
    andrzej.kaczmarczyk@cvut.cz, schiesim@fit.cvut.cz,}}\\
		{\small\texttt{
    nicholas.axel.135@gmail.com, haifengxu@uchicago.edu}}
}

\usepackage{marginnote}

\usetikzlibrary{shadows}
\usetikzlibrary{shapes.geometric} %

\newcommand{\tagStar}{%
\tikz[baseline=(n.base)]{
\node[
  circle,
  minimum size=20pt,
] (circ) {};
\node[
  star,
  star points=5,
  star point ratio=2,
  fill=black!50,
  inner sep=1pt,
  minimum size=8pt
] (n) {};
}%
}

\newcommand{\tagStarInline}[1]{%
\tikz[baseline=(n.base)]{
\node[
  star,
  star points=5,
  star point ratio=2,
  fill=black!50,
  inner sep=1pt,
  minimum size=8pt,
] (n) {};
}%
}

\newcommand{\toclaimtag}{%
\tikz[baseline=(n.base)]{
\node[
  circle,
  minimum size=20pt,
  font=\small,  
  text=black!50,                   
] (n) {\faIcon{arrow-alt-circle-up}};
}%
}

\newif\ifshowdeferredtag
\showdeferredtagtrue

\newcommand{\deferredproofmark}[1]{%
  \hypertarget{thm:#1}{}%
  \ifshowdeferredtag
    \begingroup
    \setlength{\marginparsep}{0pt}
    \reversemarginpar
    \marginnote{%
      \hyperlink{proof:#1}{\tagStar{}}%
    }[-5pt]%
    \endgroup
  \fi%
}

\newcommand{\backtotheorem}[1]{%
  \hypertarget{proof:#1}{}%
  \begingroup
  \setlength{\marginparsep}{0pt}
  \reversemarginpar
  \marginnote{%
    \hyperlink{thm:#1}{\toclaimtag{}}%
  }[-6pt]%
  \endgroup
}

\newcommand{\restatetheorem}[1]{%
  {\showdeferredtagfalse #1}%
}

\begin{document}

\maketitle

\begin{abstract}
    Coordinating agents through hazardous environments, such as aid-delivering drones navigating conflict zones or field robots traversing deployment areas filled with obstacles, poses fundamental planning challenges. We introduce and analyze the computational complexity of a new multi-agent path planning problem that captures this setting. A group of identical agents begins at a common start location and must navigate a graph-based environment to reach a common target. The graph contains hazards that eliminate agents upon contact but then enter a known cooldown period before reactivating. In this discrete-time, fully-observable, deterministic setting, the planning task is to compute a movement schedule that maximizes the number of agents reaching the target. We first prove that, despite the exponentially large space of feasible plans, optimal plans require only polynomially-many steps, establishing membership in \NP. We then show that the problem is \NPh even when the environment graph is a tree. On the positive side, we present a polynomial-time algorithm for graphs consisting of vertex-disjoint paths from start to target. Our results establish a rich computational landscape for this problem, identifying both intractable and tractable fragments.
\end{abstract}

\section{Introduction}

A humanitarian aid organization is planning to transport life-critical supplies
through a conflict zone. The hostile actions of parties involved may severely
constrain the planning task: these parties may attempt to divert or seize aid
shipments. In particular, hostile forces monitor the locations under their
control. When a convoy crosses such a location, it is detected and captured.
However, capturing a convoy occupies the hostile patrol for a predictable period during which subsequent convoys may pass safely. How should
the organization schedule dispatches to maximize the number of convoys reaching
their destination?

This scenario motivates a more general planning problem. A central planner
manages a set of identical agents at a common starting location. They seek to
navigate as many of these agents as possible through a hostile environment to a
designated target, potentially sacrificing some agents to enable others to pass.
Beyond humanitarian logistics~\citep{Wassenhove06,RottkemperF2013}, this model
captures applications such as cargo transport through pirate-prone
waters~\citep{JakobVHP2012}, UAV-based delivery of resources~\citep{WuRC2016},
and analysis of cybersecurity intrusion detection where probe packets cause a
temporary detector failure~\citep{pta-new:r:eluding-IDS}.

We formalize the environment as a graph whose vertices represent locations and
whose edges represent feasible transitions. Two vertices are designated as the
start and target. Hostility is modeled via traps placed at certain vertices.
Each trap eliminates any agent that enters its vertex but then enters a reload
period of known duration before reactivating. The planning task is to compute a
movement schedule that maximizes the number of agents reaching the target. This
formulation yields a multi-agent path planning problem with a new
temporal-strategic dimension.

A rich literature spans related topics including constrained path
finding~\citep{PhillipsL2011,AhmadiTHK2021},
multi-agent path finding~\citep{SternSSFKMWLACSBB2019}, token
reconfiguration~\citep{CalinescuDP2008}, pursuit-evasion games~\citep{BorieTK11},
and
humanitarian logistics in conflict zones~\citep{BoehmerHXT2024} (see
\Cref{sec:relatedWork} for further discussion). However, to the best of our
knowledge, the existing models do not address our specific problem.
The key distinction lies in the nature of our traps, which are statically placed
(limiting the applicability of game-theoretic frameworks), deterministic (making
stochastic models unnecessarily general), and, crucially, which require a
reload time after each capture. This reload mechanic is new and demands
novel analytical techniques.
 
We briefly justify our modeling choices. First, reloading of traps captures
deceptive strategies in which early agents absorb hostile attention allowing
those who follow to pass. Second, we model a static rather than
adversarial environment. Instead of analyzing repeated games, we focus on a
single, short-horizon assignment within a known environment; an essential
building block for longer-horizon decision-making. Third, while stochastic
models may appear more realistic, they are notoriously difficult to solve in
settings like ours. Furthermore, our deterministic approach isolates the fundamental
computational structure of the problem and admits a natural interpretation as
worst-case reasoning by a risk-averse planner. Finally, we assume that the
reload durations are known. In practice, these can be estimated from
historical data or
knowledge of standard operational procedures.

\paragraph{Our Contribution.}

We introduce a novel multi-agent path-planning model for transportation in hostile environments and analyze it from the perspective of computational complexity.
First, we identify tractable fragments of the problem. We show that if the underlying topology is a simple path, then an optimal strategy can be computed in polynomial time. Despite the simplicity of this topology, our \emph{run-wait-sacrifice} strategy requires surprisingly non-trivial arguments about the structure of optimal plans. We then extend these results to more general topologies whose condensation (see \Cref{def:condensation}) has a linear structure; 
these include, e.g., vertex-disjoint unions of paths.
In contrast to these positive results, we establish computational 
hardness for modest generalizations. Specifically, we prove that the 
problem is \NP-hard even when the underlying topology is a tree with 
maximum degree~$3$, or when each agent crosses at most $6$ traps.

For the sake of readability, we deferred some proofs to~\Cref{app:proofs}. These
are marked with~\raisebox{3pt}{\tagStarInline{}}, a symbol linking to the
corresponding proofs.

\section{Related Work}\label{sec:relatedWork}

Our problem is connected to various research streams in computer science, 
mathematics, operations research, artificial intelligence, and planning %
literature.

\paragraph{Classical Path Finding.}
The foundational problem underlying our model is \probName{Reachability}: given
a graph~$G$ and vertices~$s$ and~$t$, determine whether~$t$ can be reached
from~$s$. This can be decided in linear time via a depth-first search, and the
shortest paths can be computed in polynomial time using Dijkstra's
algorithm~\citep{Dijkstra1959} or, with heuristic guidance, the A*
algorithm~\citep{HartNR1968}. However, our setting differs fundamentally: the
presence of traps with reload periods introduces temporal dynamics, and we must
coordinate paths for multiple agents rather than a single one. Substantial work
has also addressed reachability under
\emph{uncertainty}~\citep{WagnerC2017,ShoferSS2023}, in \emph{distributed}
settings~\citep{CapVK2015}, and in \emph{faulty} environments%
~\citep{PapadimitriouY1991,FriedSBW2013,FioravantesKMO2025}; in contrast, we
assume a fully-informed central planner and deterministic, fault-free
environment.

\paragraph{Dynamic and Temporal Environment.}
Recent research has explored path finding in graphs whose structure changes over
time. In \emph{temporal graphs}, edge availability is restricted to specific
discrete time steps. While single-agent reachability remains polynomial-time
solvable in temporal graphs~\citep{WuCKHHW2016}, multi-agent variants become
computationally hard even on simple paths~\citep{KlobasMMNZ2021}. In a different
direction, \citet{CarmesinWPKM2023} and \citet{DvorakKOPSS2025} studied
\emph{self-deleting graphs}, where visiting a vertex removes certain incident
edges. Our model differs from both: all edges remain available at all times, but
agents can be eliminated by traps that subsequently enter a reload period.

\paragraph{Multi-Agent Path Finding.}
Our work is closely related to the extensively studied \probName{Multi-Agent
Path Finding} (MAPF)
problem~\citep{FelnerSSBGSSWS2017,SternSSFKMWLACSBB2019,WangXZLLWL2025}. In
MAPF, each of~$k$ agents has a designated start and goal location, and the
objective is to find collision-free paths minimizing makespan or total travel
time. MAPF has been thoroughly analyzed from a computational complexity
perspective~\citep{Surynek2010,YuL2013,Nebel2024,fioravantesExactAlgorithmsLowerbounds2023,FioravantesKKMOV2025a,DeligkasEGK025},
and numerous algorithmic approaches have been developed~\citep{SharonSFS2015,LiHS0K2019,Surynek2022}. Our problem differs in several key respects: (i)~all agents share a common start and target, (ii)~certain vertices contain traps that eliminate agents, and (iii)~the objective is to maximize the number of surviving agents rather than minimizing travel time. These differences fundamentally change the problem structure.

\paragraph{Pursuit-Evasion Games.}
\emph{Pursuit-evasion games}~\citep{Parsons1976,BorieTK11} study scenarios where
pursuers attempt to capture evaders in a graph. The complexity of these problems
depends heavily on the graph class and movement rules; for instance, determining
the minimum number of pursuers needed to clear a graph is related to treewidth
and pathwidth~\citep{SeymourT1993}. While superficially similar, our model
differs substantially. Our traps are statically placed rather than actively
pursuing agents, they operate deterministically rather than strategically, and
crucially, they need to reload after each capture---a feature not considered 
in classical pursuit-evasion problems.

\paragraph{Token Reconfiguration.}
In \emph{token reconfiguration}
problems~\citep{KornhauserMS1984,HearnD2005,CalinescuDP2008,Heuvel2013}, tokens
placed on graph vertices must be moved from an initial configuration to a target
configuration under various movement rules (sliding along edges or jumping to
non-adjacent vertices). These problems are typically \NPc and \APXh on general
graphs, though polynomial-time algorithms exist, e.g.,\,for
trees~\citep{CalinescuDP2008,DemaineDFHIOOUY2015}. 
While our agents sometimes resemble moving tokens (and we use this analogy in
some proofs), the key
distinction is that our agents may be \emph{eliminated} during transit, and the
objective is to maximize the number of agents reaching the target.%

\paragraph{Video Games and Motion Planning.}
A surprising connection exists between our work and certain video games. For
example, in \emph{Lemmings}~\citep{McCarthy1998}, agents spawn from a door and
must be guided to a target, navigating hazards that become temporarily passable
after eliminating an agent. The key difference is that Lemmings agents move
autonomously and are controlled only indirectly via skill assignments, whereas
we directly control agent movements. Moreover, Lemmings solutions can require
exponentially many steps~\citep{Forisek2010}, and the decision problem is
\PSc~\citep{Viglietta2015}; in contrast, we show that optimal plans in our model have polynomial length.

\paragraph{Security Games.} Finally, there are several game-theoretical models
of games related to security applications, like Stackelberg~\citep{Tambe2012} or
Blotto~\citep{BehnezhadBDHMPR2018,Kazmierowski2025} games. The closest to ours
are the recently introduced \emph{escape sensing games}~\citep{BoehmerHXT2024},
where one player tries to move as many of her assets from the start to target
location, while an adversary uses sensors to detect some of the moving assets.
In contrast, in our work the behavior of traps is implicit from their reload
times.

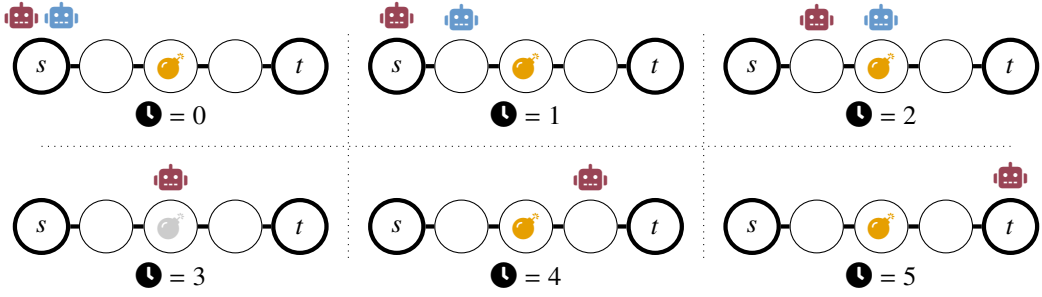
\begin{figure*}
    \centering
    \begin{tikzpicture}[every node/.style={minimum width=7mm},scale=0.85]
        \draw[dotted, thin] (0,-1.25) -- (15, -1.25);
        \draw[dotted, thin] (4.75,.5) -- (4.75, -3);
        \draw[dotted, thin] (10.25,.5) -- (10.25, -3);

        \node[draw,circle,ultra thick,label=90:{\textcolor{cbRed}{\faIcon{robot}}~\textcolor{cbBlue}{\faIcon{robot}}}] (s) at (0,0) {$\sV$};
        \node[draw,circle] (v2) at (1,0) {};
        \node[draw,circle,inner sep=1pt] (v3) at (2,0) {\textcolor{cbOrange}{\faIcon{bomb}}};
        \node[draw,circle] (v4) at (3,0) {};
        \node[draw,circle,ultra thick] (t) at (4,0) {$\tV$}; 

        \node[] at (2,-0.75) {$\text{\faIcon{clock}}=0$};
        \draw[ultra thick] (s) -- (v2) -- (v3) -- (v4) -- (t);

        \node[draw,circle,ultra thick,label=90:{\textcolor{cbRed}{\faIcon{robot}}}] (s) at (5.5,0) {$\sV$};
        \node[draw,circle,label=90:{\textcolor{cbBlue}{\faIcon{robot}}}] (v2) at (6.5,0) {};
        \node[draw,circle,inner sep=1pt] (v3) at (7.5,0) {\textcolor{cbOrange}{\faIcon{bomb}}};
        \node[draw,circle] (v4) at (8.5,0) {};
        \node[draw,circle,ultra thick] (t) at (9.5,0) {$\tV$}; 

        \node[] at (7.5,-0.75) {$\text{\faIcon{clock}}=1$};
        \draw[ultra thick] (s) -- (v2) -- (v3) -- (v4) -- (t);

        \node[draw,circle,ultra thick] (s) at (11,0) {$\sV$};
        \node[draw,circle,label=90:{\textcolor{cbRed}{\faIcon{robot}}}] (v2) at (12,0) {};
        \node[draw,circle,inner sep=1pt,label=90:{\textcolor{cbBlue}{\faIcon{robot}}}] (v3) at (13,0) {\textcolor{cbOrange}{\faIcon{bomb}}};
        \node[draw,circle] (v4) at (14,0) {};
        \node[draw,circle,ultra thick,label=90:{\phantom{\faIcon{robot}~\faIcon{robot}}}] (t) at (15,0) {$\tV$}; 

        \node[] at (13,-0.75) {$\text{\faIcon{clock}}=2$};
        \draw[ultra thick] (s) -- (v2) -- (v3) -- (v4) -- (t);

        \node[draw,circle,ultra thick] (s) at (0,-2.5) {$\sV$};
        \node[draw,circle] (v2) at (1,-2.5) {};
        \node[draw,circle,inner sep=1pt,label=90:{\textcolor{cbRed}{\faIcon{robot}}}] (v3) at (2,-2.5) {\textcolor{gray!40}{\faIcon{bomb}}};
        \node[draw,circle] (v4) at (3,-2.5) {};
        \node[draw,circle,ultra thick] (t) at (4,-2.5) {$\tV$}; 

        \node[] at (2,-3.25) {$\text{\faIcon{clock}}=3$};
        \draw[ultra thick] (s) -- (v2) -- (v3) -- (v4) -- (t);

        \node[draw,circle,ultra thick] (s) at (5.5,-2.5) {$\sV$};
        \node[draw,circle] (v2) at (6.5,-2.5) {};
        \node[draw,circle,inner sep=1pt] (v3) at (7.5,-2.5) {\textcolor{cbOrange}{\faIcon{bomb}}};
        \node[draw,circle,label=90:{\textcolor{cbRed}{\faIcon{robot}}}] (v4) at (8.5,-2.5) {};
        \node[draw,circle,ultra thick] (t) at (9.5,-2.5) {$\tV$}; 

        \node[] at (7.5,-3.25) {$\text{\faIcon{clock}}=4$};
        \draw[ultra thick] (s) -- (v2) -- (v3) -- (v4) -- (t);

        \node[draw,circle,ultra thick] (s) at (11,-2.5) {$\sV$};
        \node[draw,circle] (v2) at (12,-2.5) {};
        \node[draw,circle,inner sep=1pt] (v3) at (13,-2.5) {\textcolor{cbOrange}{\faIcon{bomb}}};
        \node[draw,circle] (v4) at (14,-2.5) {};
        \node[draw,circle,ultra thick,label=90:{\textcolor{cbRed}{\faIcon{robot}}}] (t) at (15,-2.5) {$\tV$}; 

        \node[] at (13,-3.25) {$\text{\faIcon{clock}}=5$};
        \draw[ultra thick] (s) -- (v2) -- (v3) -- (v4) -- (t);
    \end{tikzpicture}
    \caption{An illustration of our problem showing timesteps $0$ to~$5$,
      as indicated by~{\text{\faIcon{clock}}}. There are two assets,
      \textcolor{cbRed}{\bfseries red} and \textcolor{cbBlue}{\bfseries blue}.
      The reload time of the single trap $r$ (\textcolor{cbOrange}{\faIcon{bomb}})
    is $\reloadFn(r) = 1$.  If the trap is depicted in light gray
  (\textcolor{gray!25}{\faIcon{bomb}}), then it is inactive in this round.}
    \label{fig:example}
\end{figure*}

\section{The Model}

We model the space of possible moves as an undirected graph~$G=(V,E)$ that we
refer to as \emph{topology}. Herein, the vertices represent certain important
locations\footnote{We use the term location and vertex interchangeably.} and the
edges are shortest connections between them. We assume that traversing an edge
takes one unit of time. There are two distinguished
locations~$\{\sV,\tV\} \subseteq V$\,---\,the \emph{initial} and \emph{target location},
respectively. We use~$\numVertices$
to denote
the number of vertices
of~$G$. For some~$V' \subseteq V$, we use
$G[V']$ to denote the subgraph of~$G$ induced by~$V'$. For some positive
integer~$x$, $[x]$ represents the set~$\{1, 2, \ldots, x\}$.

We assume a set~$\agents$ of~$\numAgents$ \emph{assets}\footnote{We prefer
``assets'' over ``agents,'' as the former is more general as it encompasses both
autonomous and non-autonomous entities.}
$a_1,\ldots,a_\numAgents$, which are placed at the initial
location~$\sV$ at the beginning. Our goal is to move as many assets
as possible to the target location~$\tV$.\footnote{We do not
consider other goals, such as the total time of all assets reaching the
target location.} The assets move in discrete rounds using edges of the graph,
and no pair of assets can occupy the same vertex (except for~$\sV$ and~$\tV$) on
the same round. Moreover, a subset of locations~$\traps\subseteq
V\setminus\{\sV,\tV\}$ are protected by \emph{traps}~$r_1,\ldots,r_\numTraps$.
These traps are \emph{known} in advance, and each trap~$r_i$,~$i\in[\numTraps]$, is associated with its \emph{reload time}~$\reloadFn(r_i)\in\N$.
Initially, each trap is \emph{active}. If an asset enters a location protected
by some active trap~$r_i$, then the asset is eliminated, and~$r_i$ becomes
\emph{inactive}. If~$r_i$ is inactive, then assets can pass it without being
eliminated. Trap~$r_i$ remains inactive for~$\reloadFn(r_i)$ rounds after it was
triggered, and then becomes active again. 

\begin{example}\label{ex:problemDef}
For an illustration of the problem, see the instance in~\Cref{fig:example}.
There are two assets, red and blue, and a single trap with a reload time of
one. Initially, both assets are placed on the initial vertex~$\sV$. In the first
round, the blue asset moves to the single neighbor of~$\sV$ and the red asset
has to stay on~$\sV$, as otherwise they would occupy the same vertex, which is
forbidden. In the next round, both assets move one step toward $\tV$. The blue asset
activates the trap, is eliminated, and the trap becomes inactive. Since the trap
is inactive, the red asset can move to its location in the next round without
being eliminated. However, in round~$4$, the red asset must leave this
location, as the trap gets active again and would destroy the red asset.
Finally, the red asset moves to the target location.
\end{example}

The movement of assets is represented by a \emph{movement plan}, or simply
`\emph{plan}.' Formally, introducing a special
location~$\destroyed\not\in V$ for the eliminated assets, a plan~$\plan\colon
\agents\times\N_0\to V\cup \{\destroyed\}$ assigns every asset
$a\in\agents$ and timestep~$j\in\N_0$ the location~$v\in V \cup \destroyed$
of~$a$ at time~$j$. The plan~$\plan$ is \emph{valid} if all the following
conditions are satisfied:
\begin{enumerate}
    \item $\plan(a,0) = \sV$ for every~$a\in\agents$ and, for every~$j > 0$,
      if~$\plan(a,j) = \sV$ then~$\plan(a,j-1) = \sV$; that is, every asset
      starts at the initial location~$\sV$ and never returns there after leaving
      once,
    \item if~$\plan(a,j) \in \{\tV,\destroyed\}$ for some~$a\in\agents$
      and~$j\in\N_0$, then~$\plan(a,j+1) = \plan(a,j)$; that is, assets cannot
      leave the target and special location after they reach
      them,\label{prop:planNeverLeavesTargetOrDestroyed}
    \item for every~$a\in\agents$ and every~$j\in\N$ it holds that
      if~$\plan(a,j) \not= \destroyed$, then~$\plan(a,j)$ is adjacent to or
      equal to $\plan(a,j-1)$; that is, at each round assets move along the
      edges of the topology or stay put,
    \item for two distinct assets~$a,a'\in\agents$ and~$j\in\N_0$,~$\plan(a,j) =
      \plan(a',j)$ if and only if~$\plan(a,j) \in \{\sV,\tV,\destroyed\}$; that
      is, no two assets can stand on the same vertex~$v$, unless~$v$ is the
      initial, target, or special location $\destroyed$,
    \item if~$\plan(a,j) = r \in \traps$ for some~$a\in\agents$ and~$j\in\N$ and
      there is no~$a'\in\agents$ and~$\ell\in[\reloadFn(r)]$ such
      that~$\plan(a',\max\{0,j-\ell\}) = r$, then~$\plan(a,j+1) = \destroyed$;
      that is, if an asset steps on an active trap, it is eliminated, %
    \item for every~$j\in\N$, if there is~$a\in\agents$ such
      that~$\plan(a,j)\not\in\{\tV,\destroyed\}$, then for at least
      one~$a'\in\agents$ we have~$\plan(a',j)\not=\plan(a',j+1)$; that is, as
      long as some asset remains, at least one asset moves each round, and
    \item there is some~$\len(\plan) \in \N$ such that for every~$a\in\agents$
      it holds that~$\plan(a,\len(\plan)) \in \{\tV,\destroyed\}$, and there is
      an asset~$a\in\agents$
      with~$\plan(a,\len(\plan)-1)\not\in\{\tV,\destroyed\}$; that is, all
      assets either reach the target or are eliminated in finite time.
\end{enumerate}
Unless stated explicitly, we assume only valid plans. For a valid plan~\plan{}
of \emph{length}~$\len(\plan)$, we say that the plan~\emph{ends}
at~$\len(\plan)$, and we say that~$a \in \agents$ \emph{survived} if
${\plan(a,\len(\plan)) = \tV}$; otherwise,~$a$ \emph{failed}.

\begin{example}\label{ex:movementPlan}
    A plan capturing the movement of the assets as described in
    \Cref{ex:problemDef} is as follows (we assume that the internal vertices
    of~$G$ are, from left to right,~$v_2$,~$v_3$, and~$v_4$).
    \begin{center}
        \begin{tabular}{c|c|c|c|c|c|c}
           $\plan$ & 0 & 1 & 2 & 3 & 4 & 5  \\
             \midrule
           \textcolor{cbBlue}{\faIcon{robot}}  
             &~$\sV$ &~$v_2$ &~$v_3$ &~$\destroyed$ &~$\destroyed$ &~$\destroyed$ \\
           \textcolor{cbRed}{\faIcon{robot}}   
             &~$\sV$ &~$\sV$ &~$v_2$ &~$v_3$ &~$v_4$ &~$\tV$ \\
        \end{tabular}
    \end{center}
    It is easy to verify that this plan is valid and  has length~${\len(\plan) = 5}$.
\end{example}

We study the above-defined model through the lens of computational complexity of
the following decision problem.

\defProblemQuestion{\basicProb}%
    {A set of assets~$\agents$, a topology~$G = (V,E)$, distinct initial and
    target vertices~$\sV,\tV \in V$, a set of traps~$\traps\subseteq
    V\setminus\{\sV,\tV\}$, a reload time function~$\reloadFn\colon \traps\to
   \naturals$, and a goal~$\goal \in [|\agents|]$.}%
    {Is there a valid plan~$\plan$ yielding at least~$\goal$ surviving assets?}

\begin{remark}
    Unlike in (some variants of) \probName{Multi-Agent Path Finding}, we do not explicitly forbid two assets from swapping positions across an edge. However, since our assets are indistinguishable, any swap of two assets $a_i$ and $a_j$ in a plan~$\plan$ can be eliminated by the following transformation:
    \begin{enumerate*}[label=(\roman*)]
        \item At the time of the swap, let both assets wait at their respective pre-swap positions (i.e., before they would have crossed the edge).
        \item After that time, assign to $a_i$ the remainder of $a_j$'s plan, and to~$a_j$ the remainder of $a_i$'s plan.
    \end{enumerate*}
    By exhaustively applying this transformation to all swaps in $\plan$, one obtains a modified plan without any swaps, achieving the same number of surviving assets, with identical temporal performance. Thus, when the application requires a plan without swaps, we may enforce this restriction without loss of optimality.
\end{remark}

\section{The Existence of Short Plans}

The properties required in the definition of a valid plan directly forbid some
unreasonable solutions, such as infinite plans or keeping all assets at the
initial location without moving any of them. However, it is not strong enough to
forbid very lengthy plans, e.g., plans whose length is exponential in the input
size. The following theorem shows that a plan~$\plan$ ensuring that
$\goal$~assets survive implies a \emph{short} plan~$\plan'$ with at least the
same number of surviving assets.

\begin{theorem}\label{thm:shortPlans}
  Let~$\mathcal{I} = (G, \sV, \tV, \traps, \reloadFn, \agents, \goal)$ be an
  instance of \basicProb{} and let~$\plan$ be a plan of length~$\ell$
  with at least~$\goal$ surviving assets witnessing that~$\mathcal{I}$ is a
  \Yes-instance. Then, there is a plan~$\plan'$ of
  length~$\ell'\in\operatorname{poly}(n,\numAgents)$ with at least~$\goal$
  surviving assets. 
\end{theorem}
\begin{proof}
  Our first step is to transform the given network~$G$ into a vertex-colored
  graph that depends on~$\plan$. Given the new colored graph, we interpret
  movement plans as graph-coloring reconfigurations and combine existing results
  with new observations to obtain our theorem.
  \begin{definition}
    A \emph{simple transform} of a network~$G=(V,E)$ and a number~$\goal$ of
    surviving assets for some plan, is a graph~$\hat{G}=(V \cup S \cup T, E \cup
    E')$ such that vertex set~$S$ consists of~$\numAgents - 1$~vertices arranged
    in a path attached at one of its ends to~$\sV \in V$, vertex set~$T$ consists
    of~$\goal-1$~vertices arranged in a (different) path attached with one of
    its ends to~$\tV$, and~$E'$ contains the respective edges.
  \end{definition} 

  \noindent
  To fully exploit the simple transform and capture the dynamics of moving
  assets, we define a specific graph coloring.

  \begin{definition}\label{def:snapshot}
    For some timestep~$i \in \naturals$, an \emph{$i$-snapshot} of a
    network~$G=(V,E)$ and a plan~$\plan$ with length at least~$i$ is
    tuple~$\snap^{G, \plan}_i = (\hat{G}, h)$, where~$\hat{G} = (V \cup S \cup
    T, E \cup E')$ is the simple transform of~$G$ and~$\plan$ and~$h$ is a
    binary coloring of the vertices of~$\hat{G}$ such that:
    \begin{enumerate}
      \item for each~$v \in V$,~$h(v) = 1$ if there is an asset~$a \in
        \agents$ that, according to~$\plan$, occupies~$v$ at time~$i \in \naturals$,
       $h(v) = 0$ otherwise;
      \item given that the number~$x$ of assets occupy the initial
        vertex~$\sV$ at time~$i$, for~$(x-1)$~vertices~$v \in S$ whose distance is
        at most~$x-1$ from~$\sV$,~$h(v) = 1$, and~$h(v) = 0$ for the remaining
        vertices in~$S$; and  
      \item given that the number~$y$ of assets occupy the terminal
        vertex~$\tV$ at time~$i$, for~$(y-1)$~vertices~$v \in T$ whose distance is
        at most~$(y-1)$ from~$\tV$,~$h(v) = 1$, and~$h(v) = 0$ for the remaining
        vertices in~$T$.  
    \end{enumerate}
    For readability, we use~$\snap_{i}$ whenever it is unambiguous.
  \end{definition} 

  Conveniently, the number of timesteps needed to reach one snapshot from
  another is upper-bounded by a polynomial.

  \begin{lemma}\label{lem:shapshots-reachability}
    Given a network~$G=(V,E)$, an empty set~$R = \emptyset$ of traps, a
    plan~$\plan$, and two snapshots~$\snap^{G, \plan}_i$ and~$\snap^{G,
    \plan}_j$ for~$i \leq j \leq \len(\plan)$, there is a plan~$\plan'$ of
    length~$\len(\plan')< (2\numAgents + \numVertices)^2$ such that~$\snap^{G,
    \plan'}_0 = \snap_i$ and~$\snap^{G, \plan'}_{\len(\plan')} = \snap_j$. 
  \end{lemma}
  \begin{claimproof}
    By definition, snapshots~$\snap^{G, \plan}_i$ and~$\snap^{G, \plan}_j$ are
    in fact two different colorings~$h_i$ and~$h_j$, both with two colors, of
    the same graph~$\hat{G} = (V \cup S \cup T, E \cup E')$, which has~$\hat{n}
    \coloneq \numAgents - 1 + \numVertices + \goal - 1 < 2\numAgents +
    \numVertices$~vertices, as~$\goal \leq \numAgents$.

    Given this interpretation, to find the claimed moving plan means to obtain
    a~$2$-coloring~$h_j$ from that of~$h_i$ via a series of certain steps.
    Specifically, in each step, we can swap colors of a pair of adjacent
    vertices of mutually distinct color. Indeed, such a step represents moving
    an asset from the vertex it occupies (colored~$1$ by definition of a
    snapshot) to some vertex that no asset occupies (colored~$0$ by the same
    definition).

    It is sufficient to only upper-bound the number of steps needed for
    transforming the colorings. To do so, we directly apply the result
    of~\citet[Lemma 1]{yam-et-al:j:swapping-colored-tokens} stating that the
    upper bound is~$\binom{\hat{n}}{2} < (2\numAgents + \numVertices)^2$, as we
    claim.
  \end{claimproof}

  We introduce a restricted snapshot that allows us to analyze snapshots of a
  certain part of the topology.
  \begin{definition}\label{def:restricted-snapshot}
    Consider some timestep~$i \in \naturals$, a network~$G=(V,E)$, a plan~$\plan$, an
    \emph{$i$-snapshot}~$\snap^{G, \plan}_i = (\hat{G}, h)$, where~$\hat{G} = (V
    \cup S \cup T, E \cup E')$, and a subset~$V' \subseteq V$.
    Let~$\widehat{V}$ be a copy of~$V'$ augmented with~$S$ if $\sV \in V'$ and
    with~$T$ if~$\tV \in V'$.
    Then, a \emph{$V'$-restricted~$i$-snapshot}~$\snap^{G, \plan}_i\rest_{V'}$
    is a tuple $(\hat{G}[\widehat{V}], h\rest_{\widehat{V}})$,
    where~$h\rest_{\widehat{V}}$ is the restriction of~$h$ to vertices
    in~$\widehat{V}$.
  \end{definition} 

  It is immediate that~\Cref{lem:shapshots-reachability} can be applied to a
  pair of snapshots restricted to a set of vertices that do not contain traps and
  form a connected component.
  \begin{corollary}\label{cor:restricted-snapshots-reachability}
    Given a network~$G=(V,E)$, a set~$R \subset V$ of traps, a plan~$\plan$, a
    subset~$V' \subseteq V$ such that~$V' \cap R = \emptyset$ and~$G[V']$ is a
    connected component, and two snapshots~$\snap^{G,
    \plan}_i\rest_{V'}$ and~$\snap^{G, \plan}_j\rest_{V'}$
    for~$i \leq j \leq \len(\plan)$ whose number of vertices colored~$1$ are the
    same, there is a plan~$\plan'$ of length~$\len(\plan')< (2\numAgents +
    \numVertices)^2$ such that~$\snap^{G, \plan'}_0\rest_{V'} = \snap^{G,
    \plan}_i\rest_{V'}$ and~$\snap^{G, \plan'}_{\len(\plan')}\rest_{V'} = \snap^{G,
    \plan}_j\rest_{V'}$.
  \end{corollary}

  We now show that contracting a part of a plan between two snapshots is
  sufficient to contract the whole plan. 
  \begin{lemma}\label{lem:shortening}
    Given an instance~$(G, \sV, \tV, \traps, \reloadFn, \agents, \goal)$ of
    \basicProb{}, let~$\plan$ be a valid plan of length~$\len(\plan) = x
    + \delta + y$ in which at least~$\goal$ assets survive such that~$\plan$
    does not activate any trap in timesteps~$x+1$ to~$x+\delta-1$, inclusive. If
    there is a (non-valid) plan~$\plan''$ of length~$\len(\plan'') = x +
    \delta''$ such that \begin{enumerate} \item for each asset~$\agent \in
      \agents$ and all~$i \in [x]$ it holds that~$\plan(\agent, i) =
      \plan''(\agent, i)$; \item~$\snap^{\plan''}_{x+\delta''} =
      \snap^{\plan}_{x+\delta}$; \item plan~$\plan''$ does not activate any trap
      in timesteps~$x+1$
        to~$x+\delta''-1$ inclusive; and
      \item~$\delta'' < \delta$;
    \end{enumerate}
    then there exists a valid plan~$\plan'$ of length~$\len(\plan') = x +
    \delta'' + y$ with at least~$k$ surviving assets.
  \end{lemma}
  \begin{claimproof}
    We build plan~$\plan'$ by copying~$\plan''$ and then extending it with the
    last~$y$~steps of~$\plan$, adapting them if needed.

    Consider running~$\plan'$ according to~$\plan''$ until timestep~$x +
    \delta''$. Since~$\snap^{\plan'}_{x+\delta''} = \snap^{\plan}_{x+\delta}$,
    we now attempt to replicate in~$\plan'$ every timestep~$j$, in the natural
    order, from~$x + \delta + 1$ to~$x+\delta + y$ according to plan~$\plan$.
    Note that we certainly succeed for the timestep~$x + \delta + 1$
    because~$\snap^{\plan}_{x+\delta} = \snap^{\plan'}_{x+\delta''}$.

    For some $j$, the replication might become impossible. That is, there is an
    asset~$\agent$ such that at timestep~$j$ of~$\plan$
    this asset should occupy a vertex~$v$ occupied by another asset~$\agent'$.
    Our intention is to leave~$\agent$ intact at vertex~$v$ at timestep~$j$ but
    this, however, might invalidate the move of another asset at the same time,
    the one that potentially was supposed to take position at~$v$ instead
    of~$\agent$. Hence, we collect those assets that we cannot move due to
    leaving~$\agent$ intact. On the graph, these assets can form either a path
    or a cycle. In both cases, we leave \emph{all} of these assets intact at
    timestep~$j$. Since we are considering the first such~$j$ for which
    mimicking~$\plan$ failed, the described situation can happen only because of
    some asset~$\agent'$ that was destroyed at time~$j$. This means, in
    particular, that~$\agent'$ moves at time~$j$ to the destroyed state,
    $\destroyed$, in plan~$\plan$. However, we also let~$\agent'$ stay intact,
    as it is now not destroyed; for, in the opposite case~$\agent$ would
    have been able to take its place. It is not hard to verify that because of
    these alignments, we obtain a snapshot~$\snap^{\plan'}_{j}$ whose set~$O$ of
    vertices colored~$1$ is a superset (potentially a strict superset) of that
    of~$\snap^{\plan}_{j}$. Put in words, the former contains all assets at the
    same vertices as that of the latter, and, second, the former has maybe even
    more assets. For the next step, we relabel the assets in our plan~$\plan$
    such that they meet exactly the positions of vertices in our just-extended
    plan~$\plan'$. There will be assets that cannot be matched, and these will
    be deemed to stay put in~$\plan'$, unless another conflict like the
    described one appears.

    As a result of running the whole ``copying'' procedure, we obtain~$\plan'$
    that loses at most as many assets as the original plan~$\plan$. Indeed, our
    procedure of copying never decreased the number of non-destroyed assets.
    Potentially, the remaining vertices might be directed to the final vertex,
    if that turns out to be possible. 
  \end{claimproof}

  We combine the previous observations into the following crucial lemma, which
  upper-bounds the number of timesteps between two subsequent activations of
  traps.

  \begin{lemma}\label{lem:global-shortening}
    Let~$k_1, \ldots, k_q$, such that~$k_i \leq k_{i+1}$ for each~$i \in
    [q]$ be the timesteps at which plan~$\plan$ activates traps. Then, if there
    is~$i \in [q-1]$ such that~$k_{i+1} - k_i > (2\numAgents + \numVertices)^2$,
    then there is a plan~$\plan'$ such that~$\snap^{G, \plan'}_0 = \snap^{G,
    \plan}_0$ and~$\snap^{G, \plan'}_{\len(\plan')} = \snap^{G,
    \plan}_{\len(\plan)}$ and such that in the trap activation series~$k'_1,
     \ldots, k'_q$ of~$\plan'$ it holds that~$k'_{i+1} - k'_i <
    (2\numAgents + \numVertices)^2$.
  \end{lemma}
  \begin{claimproof}
    Let us fix~$i$ such that~$k_{i+1} - k_i > (2\numAgents + \numVertices)^2$,
    and let~$K \subseteq V$ be the set of vertices passed by at least one asset
    between timesteps~$k_{i}$ and~$k_{i+1}$, inclusive. Now, let us consider the
    graph~$G' = G[V \setminus \traps \cup K]$. Note that the vertices of~$G'$
    are either non-trap vertices or vertices whose trap is deactivated between
    timesteps~$k_i$ and~$k_{i+1}$, so they can be considered non-trap vertices
    during that period.
    
    Let~$C_1, C_2, \ldots C_x$  be the connected components of~$G'$ (clearly,
    there must be at least one). According
    to~\Cref{cor:restricted-snapshots-reachability}, for each component~$C_j$,
    $j \in [x]$, there is a plan~$\plan_j$ such that~$\snap_{k_i +
    1}^{\plan}\rest_{C_j} = \snap_{k_i + 1}^{\plan_j}\rest_{C_j}$
    and~$\snap_{k_{i+1} - 1}^{\plan}\rest_{C_j} = \snap_{k_{i+1} -
    1}^{\plan_j}\rest_{C_j}$.  Let~$L$ be the maximum length among these new
    plans~$\plan_1, \plan_2, \ldots, \plan_x$. We build a new plan~$\plan'$ as
    follows. We start with copying the plan~$\plan$ for the first~$k_i$
    timesteps for all assets. Then, we let each group~$A_j$ of assets that
    remained in the respective connected component~$C_j$ at timesteps~$k_i + 1$
    to~$k_{i+1}-1$ follow their plan~$\plan_j$ for a~$\len(\plan_j)$~timesteps.
    For each such~$j$ that~$\len(\plan_j) < L$, we let the respective assets
    from group~$A_j$ stay put for the next~$L - \len(\plan_j)$~timesteps.
    Finally, we again let all assets follow the original plan~$\plan$, which
    completes the construction of~$\plan'$.

    Naturally,~$\snap^{G, \plan'}_{y} = \snap^{G, \plan}_{y}$ for each~$y \in [0,
    k_{i}]$ as the respective plans are identical at these timesteps. By
    applying~\Cref{cor:restricted-snapshots-reachability} to plans~$\plan_1,
    \ldots, \plan_x$, copying these plans into~$\plan'$, and enforcing
    the same length of~$L$ by forced idling, it holds that~$\snap^{G,
    \plan'}_{k_i + L} = \snap^{G, \plan}_{k_{i+1} - 1}$.  Note that by the fact
    that the vertices of~$G'$ are either non-trap vertices or vertices for which
    their trap is deactivated between timesteps~$k_i$ and~$k_{i+1}$, applying
    the forced idling does not activate any trap.
    Finally, for the remaining part of plans~$\plan$ and~$\plan'$, we note that
    they meet the conditions of~\Cref{lem:shortening}. So, we apply this lemma
    and get that
    in plan~$\plan'$ there are at least as many surviving assets as in
    plan~$\plan$.
  \end{claimproof}

  To get the result, we exhaustively apply~\Cref{lem:global-shortening}.
\end{proof}

Due to the previous theorem, we can focus only on deciding whether there is a
short solution to the \basicProb problem. Hence, for the rest of the paper, we
assume only short solutions without explicitly specifying it. Therefore, we
obtain the following.

\begin{theorem}\label{cor:probInNP}
    \basicProb is in \NP.
\end{theorem}
\begin{proof}
    By \Cref{thm:shortPlans}, if a given instance~$\mathcal{I}$ is a
    \Yes-instance, then there exists a certificate (plan~$\plan$) of polynomial
    length. Therefore, to verify the certificate, we simulate it, check its
    validity at every step, and observe whether at least~$\goal$ assets have been
    successful.
\end{proof}

\section{Efficient Algorithms}\label{sec:algos}

The containment in \NP{} established in the above section is just the first step
on the way of understanding the hardness of~\basicProb{}. In particular, it does
not imply efficient solvability, a trait potentially of practical relevance.
So, can we do better and solve~\basicProb{} in polynomial time?

Chasing the answer to the posed question, let us first make several helpful
assumptions on the input of~\basicProb{}. By this, we will avoid distracting
trivial cases and simplify the formal analysis of the computational complexity
of our problem. To this end, let us fix some instance~$\mathcal{I} = (G, \sV,
\tV, \traps, \reloadFn, \agents, \goal)$ of~\basicProb{}. Then, we assume
without loss of generality that:
\begin{enumerate}
  \item $G$ contains an~$\sV$-$\tV$~path, as otherwise~$\mathcal{I}$ is
    trivially a \No-instance;
  \item $G$ is connected; if not, using the breadth first search, we identify in
    polynomial time all components of~$G$ and narrow the topology down to
    the one containing~$\sV$ and~$\tV$;
  \item $G[V \setminus \traps]$ does not contain an~$\sV$-$\tV$ path;
    otherwise, $\mathcal{I}$ is a \Yes-instance because all assets
    survive by following the trap-free path in question (note that~$\goal
    \leq |\agents|$).
\end{enumerate}

Having identified and dismissed the obvious cases, our first result delivers
good news. It turns out that \basicProb{} is efficiently solvable for path
topologies by following an intuitive \emph{run-wait-sacrifice} plan. Here, each asset
always moves towards the target---runs---if only the following location is
unoccupied and is not an active trap, waits before an active trap~$r$ until it
is followed by at least~$\reloadFn(r)$ assets, and then sacrifices itself by
stepping on~$r$ and thus letting the~$\reloadFn(r)$ followers pass the
now-inactive~$r$. Contrary to the general simplicity of the run-wait-sacrifice
approach, the proof of its optimality involves a complex and careful analysis.

\begin{theorem}\label{thm:path:poly}
    If~$G$ is a path, \basicProb{} is solvable in polynomial time.
\end{theorem}
\begin{proof}
  Our algorithm is built on top of
  auxiliary~\Cref{clm:path:poly:noReturns,clm:path:poly:move,clm:path:poly:waitSacrifice}
  below
  that show what an optimal solution for path graphs looks like. Prior to
  focusing on the structure of a solution, we make several observations to get
  rid of trivial cases. Without loss of generality, we assume in the rest of the
  proof that
    \begin{enumerate*}[label=(\roman*)]
        \item $\deg(\sV) = 1$, as by property 1 of a valid plan, no asset can
          return to $\sV$ once it leaves it,
        \item $\deg(\tV) = 1$, as by property
          \ref{prop:planNeverLeavesTargetOrDestroyed} of valid plan, no asset
          can leave~$\tV$ once it reaches it.
    \end{enumerate*}
    Obviously, both properties heavily rely on the fact that the topology under consideration is a simple path. 
    As a consequence, we have that every solution plan follows the shortest $\sV,\tV$-path. For the rest of the proof, we assume that $V(G) = \{v_1,\ldots,v_\numVertices\}$, $E(G) = \{\{v_i,v_{i+1}\}\mid i\in[n-1]\}$, and $\sV = v_1$ and $\tV = v_\numVertices$.

    Previous properties show that the assets use solely vertices of the shortest $\sV,\tV$-path. However, they do not necessarily ensure that assets are not ``returning'' to previously left locations. In our first auxiliary claim, we prove that we can indeed assume only plans where, in each step, assets are either waiting on their current locations, or they are moving closer to the target location $\tV$. 

\begin{restatable}{claim}{clmPolyPathNoReturns}\label{clm:path:poly:noReturns}
  \deferredproofmark{clm:path:poly:noReturns}%
  If an instance~$\mathcal{I}$ of \basicProb, where~$G$ is a path, admits a
  solution plan~$\plan$, then it also admits a solution plan $\plan'$ such that
  for every $a\in\agents$ and every timestep $j$ we have $\plan'(a,j+1) =
  \plan'(a,j)$; $\plan'(a,j) = v_{i}$ and $\plan'(a,j+1) = v_{i+1}$ for some
  $i\in[\numVertices-1]$; or $\plan'(a,j+1) = \destroyed$.
\end{restatable}    

    Next, we show that if there is an asset that can move to an empty vertex closer to the target, that is, a vertex that is not a trap and is not occupied by another asset, this asset can always move to this vertex.

\begin{restatable}{claim}{clmPolyPathMove}\label{clm:path:poly:move}
  \deferredproofmark{clm:path:poly:move}%
  Let~$\plan$ be a solution plan, $a\in\agents$ be an asset such
  that~$\plan(a,j) = v_i$, $i\in[\numVertices-1]$, $\plan(a,j+1) \not=
  \destroyed$, $v_{i+1}$ is not an active trap in round~$j+1$, and there is no
  $a'\in\agents\setminus\{a\}$ such that $\plan(a',j+1) = v_{i+1}$.  Then, there
  is a solution plan~$\plan'$ such that~$\plan'(a,j+1) = v_{i+1}$.
\end{restatable}

That is, according to~\Cref{clm:path:poly:move}, assets
do not wait at the same vertex if they can move closer to the
target location~$\tV$. The same claim also implies that whenever an asset is
at location $v_{n-1}$, it immediately moves to $v_{n}=\tV$ in the next round
(unless it is destroyed by a trap). However, the necessary condition for the
movement of such asset is that the neighboring vertex is not a trap. In such
situation, on the other hand, we show that waiting is highly desirable.
However, for this, we need more notation. Let $r_1,\ldots,r_{\numTraps}$ be
an ordering of traps such that $\operatorname{dist}(r_i,\tV) >
\operatorname{dist}(r_{i+1},\tV)$ for every $i\in[\numTraps-1]$. We call the
sub-path between traps $r_i$ and $r_{i+1}$ a \emph{segment}, and we denote
it as $S_i$.

\begin{restatable}{claim}{clmPolyPathWaitSacrifice}\label{clm:path:poly:waitSacrifice}
  \deferredproofmark{clm:path:poly:waitSacrifice}%
  Let~$\plan$ be a solution plan, $a\in\agents$ be an asset such that
  $\plan(a,j) = v_\ell$, $\ell\in[\numVertices]$, $v_{\ell+1}\in\traps$, let the
  trap $r_{i+1}$ on $v_{\ell+1}$ be active in round $j+1$. Then, there is a
  solution plan $\plan'$ such that $a$ waits on $v_{\ell}$ and moves to
  $v_{\ell+1}$ only after all $v_{\ell-1},\ldots,v_{\ell-\ell'}$, where $\ell' =
  \min\{|S_i|,\reloadFn(r_{i+1}),|\{b\in\agents\setminus\{a\}\mid \plan(b,j) =
  v_{\ell''} \land \ell'' \leq \ell\}|\}$, are occupied by some assets.
\end{restatable}

    In other words, \Cref{clm:path:poly:waitSacrifice} says that no trap
    is activated if there are not enough assets to follow the destroyed asset.
    This rule forces us to pass traps efficiently.

    Our algorithm exploits the structure of solutions described
    by~\Cref{clm:path:poly:noReturns,clm:path:poly:move,clm:path:poly:waitSacrifice}.
    We arbitrarily fix an order of the assets, and let them leave $\sV$ in
    consecutive rounds. After the first asset passes the first trap, we
    apply the following strategy for all assets in order. Each asset moves to
    the next vertex~$v$ if (i) $v$ is unoccupied or (ii) $v$ is an active trap and
    there are enough assets to follow (according
    to~\Cref{clm:path:poly:waitSacrifice}). In all other cases, the asset in
    question stays put. Optimality and correctness of this approach follows from
    the previous claims.
\end{proof}

The run-wait-sacrifice strategy proves to be a crucial ingredient for extending
the family of polynomial-time solvable instances of~\basicProb{} with the case
of a collection of disjoint paths connecting the initial and target locations.
Intuitively, this scenario models a collection of independent routes to reach a
goal. After figuring out how many assets survive taking each of the available
routes, the decision mostly depends on how many assets to deploy to each of the
available methods. The latter task resembles making a change given a
fixed set of coin denominations, a canonical example of dynamic programming
(DP). Hence, we devise a DP-based algorithm similar to that solving the coin
change problem, using run-wait-sacrifice as a subroutine.

\begin{restatable}{proposition}{propdisjointpaths}\label{prop:disjoint-paths}
  \deferredproofmark{prop:disjoint-paths}%
  If~$G\setminus\{\sV,\tV\}$ is a disjoint union of paths, \basicProb{} can be
  solved in polynomial time.
\end{restatable}

Aiming at further exploiting the run-wait-sacrifice technique, we show that it
performs well for topology structures reaching far beyond those having an
``obviously'' linear structure (or a collection thereof). For intuition,
consider some single $\sV$-$\tV$ path. Run-wait-sacrifice achieves optimality by
waiting with passing a forthcoming trap until it ``stores'' enough assets (or as
many of them as possible) on non-trap locations between the previous and the
forthcoming trap. However, because assets are indistinguishable, the structure
of the non-trap locations used for storage is irrelevant (given that it does
contain any edge connecting to any other ``storages'') for optimality. In
particular, run-wait-sacrifice optimality retains after attaching an arbitrary
connected component of non-trap vertices to some non-trap vertex of the path.

Before formalizing the intuition above, let us introduce the notion of a
\emph{condensed topology}. Intuitively, we obtain a condensed topology from some
topology by contracting every ``connected component'' of solely non-trap (and
non-terminal) vertices of the topology into a ``supernode''
representing the contracted non-trap vertices; formally:

\begin{definition}\label{def:condensation}
    Let $G = (V,E)$ be a topology, $\traps \subseteq V$ be a collection of traps,
    $\mathcal{C} = \{C_1, C_2, \ldots, C_k\}$ be a collection of connected 
    components of graph~$G[V \setminus \traps \setminus \{\sV, \tV\}]$, where
    each~$C_i$, $i\in [k]$, has vertices~$V^\mathrm{c}_i$.
    Then, a \emph{condensed topology} is an undirected graph~$G' = (V', E')$,
    where~$V' = \{v'_i \colon C_i \in \mathcal{C}\} \cup \traps \cup \{\sV, \tV\}$
    and $E'$ contains an edge~$e = \{u, w\}$, $\{u, w\} \subset V'$, if $e \in E$,
    and contains an edge~$\{v_i, w\}$, $i \in [k]$, $w \in \traps \cup \{\sV, \tV\}$,
    if there is a vertex~$u' \in C_i$ such that ${\{u', w\} \in E}$.
\end{definition}

Now, a careful analysis of the run-wait-sacrifice strategy allows us to
extend our catalog of polynomial-time solvable instances of~\basicProb{} further to
topologies whose condensations have a linear structure.

\begin{restatable}{proposition}{propdisjoinstructure}\label{prop:disjoint-structure}
  \deferredproofmark{prop:disjoint-structure}%
  Let $\mathcal{I} = (G,\sV,\tV,\traps,\reloadFn,\agents,\goal)$ be an instance
  of \basicProb{} such that each $r \in \traps$ has degree two in~$G$ and all
  $\sV$-$\tV$~paths in the respective condensed topology are mutually disjoint.
  Then $\mathcal{I}$ is solvable in polynomial time.
\end{restatable}

As a natural next step beyond the topologies of a linear structure (or a
collection thereof), we explore a broader class of trees. We thus start with stars,
structurally very simple trees. By our initial assumptions on the topology~$G$,
the only non-trivial case to consider is that the start and target vertex are
separated with one vertex being a trap. For this case, however,
run-wait-sacrifice yields an optimal solution.\footnote{More precisely,
``run-sacrifice'' works better in this case, as the trap is a neighbor of the
initial location, so no asset needs to wait.} In passing, we mention that
precisely the considered case is the reason why we cannot
apply~\Cref{prop:disjoint-structure}, which disallows non-trap vertices
connected solely to trap vertices.

\begin{restatable}{proposition}{propstarcase}\label{prop:star-case}
  \deferredproofmark{prop:star-case}%
  If~$G$ is a star, then \basicProb{} can be solved in~$\Oh{1}$ time.
\end{restatable}

The crux of the positive results above lies mostly in exploiting the topology
structure that limits the number of sensible ways to get from the
start location to the target one.
In particular, we never took a step at a trap vertex that did not bring us
closer to the target.
As we demonstrate in the following example, this may not
be the case even for relatively small topologies.
Sometimes,
achieving optimality requires taking
a strategic non-obvious
detour from an~$\sV$-$\tV$ path, on which 
we sacrifice assets to use a certain trap vertex as ``storage.''

\begin{example}
  Consider the following instance with~$8$~assets.
  \begin{center}
    \begin{tikzpicture}[every node/.style={minimum width=5mm},x=0.75cm,
      y=0.75cm,scale=0.8, transform shape]
      \node[draw,circle,thick,label=90:{\small\textcolor{cbRed}{\faIcon{robot}}$\times
        8$}] (s) at (0,0) {$\sV$};
      \node[draw,circle,inner sep=1pt,label={[reload,label distance=1pt]-90:{\small$1$}}] (v2) at (1,0)
        {\small\textcolor{cbOrange}{\faIcon{bomb}}};
      \node[draw,circle] (v3) at (2,0) {};
      \node[draw,circle,inner sep=1pt,label={[reload,label distance=1pt]-90:{\small$2$}}] (v4) at (3,0)
        {\small\textcolor{cbOrange}{\faIcon{bomb}}};
      \node[draw,circle,inner sep=1pt,label={[reload,label distance=1pt]-90:{\small$2$}}] (v5) at (4,0)
        {\small\textcolor{cbOrange}{\faIcon{bomb}}};
      \node[draw,circle,inner sep=1pt,label={[reload,label distance=1pt]135:{\small$2$}}] (side) at (2,1)
        {\small\textcolor{cbOrange}{\faIcon{bomb}}};
      \node[draw,circle,thick] (t) at (5,0) {$\tV$}; 

      \node[above of=v5, node distance=1cm] {$\text{\faIcon{clock}}=0$};

      \draw[thick] (s) -- (v2) -- (v3) -- (v4) -- (v5) -- (t);
      \draw[thick] (v3) -- (side);
    \end{tikzpicture}\hspace{.45cm}\begin{tikzpicture}
        \draw[dotted, thin] (0,0) -- (0,-1.5);
    \end{tikzpicture}\hspace{.3cm}%
    \begin{tikzpicture}[every node/.style={minimum width=5mm},x=0.75cm,
      y=0.75cm,scale=0.8, transform shape]
      \node[draw,circle,thick,label=90:{\small\textcolor{cbRed}{\faIcon{robot}}$\times
        1$}] (s) at (0,0) {$\sV$};
      \node[draw,circle,inner sep=1pt,
        label={[reload,label distance=1pt]-90:{\small$1$}},
        label=90:{\small\textcolor{cbRed}{\faIcon{robot}}}] (v2) at (1,0)
        {\small\textcolor{cbOrange}{\faIcon{bomb}}};
      \node[draw,circle, inner sep=.5pt]
        (v3) at (2,0) {\small\textcolor{cbRed}{\faIcon{robot}}};
      \node[draw,circle,inner sep=1pt,label={[reload,label distance=1pt]-90:{\small$2$}}] (v4) at (3,0)
        {\small\textcolor{cbOrange}{\faIcon{bomb}}};
      \node[draw,circle,inner sep=1pt,label={[reload,label distance=1pt]-90:{\small$2$}}] (v5) at (4,0)
        {\small\textcolor{cbOrange}{\faIcon{bomb}}};
      \node[draw,circle,inner sep=1pt,
        label={[reload,label distance=1pt]135:{\small$2$}},
        label=90:{\small\textcolor{cbRed}{\faIcon{robot}}}]
        (side) at (2,1)
        {\small\textcolor{gray!40}{\faIcon{bomb}}};
      \node[draw,circle,thick] (t) at (5,0) {$\tV$}; 

      \node[above of=v5, node distance=1cm] {$\text{\faIcon{clock}}=7$};

      \draw[thick] (s) -- (v2) -- (v3) -- (v4) -- (v5) -- (t);
      \draw[thick] (v3) -- (side);
    \end{tikzpicture}
  \end{center}
  Ignoring the upper trap yields the case of a path, for which we know
  that run-wait-sacrifice is optimal. Simulating run-wait-sacrifice shows that
  no asset survives, as to pass the two traps before~$\tV$ we would
  need three consecutive assets; this is impossible due to the trap with reload
  time one. Yet, there is a plan yielding~$1$~successful asset achievable by
  sending all assets one by one and using the upper trap. First, one asset
  activates the upper trap at time~$5$. This unlocks the situation depicted in
  the right picture at time~$7$ that can evolve in the next round to having the
  required three consecutive assets (at the untrapped location and its sideways
  neighbors).
\end{example}

\section{Limits of Tractability}

\begin{figure*}
    \centering%
    \begin{tikzpicture}[
      every node/.style={draw,circle,minimum width=5mm,inner sep=1pt}
      ,label distance=-2mm, scale=.95]
        \tikzstyle{port} = [fill=black,minimum width=1mm]

        \node[ultra thick,minimum width=7mm] (s) at (-0.5,0) {$s$};

        \node[label={[reload, label distance=2pt]90:{\small$1$}},thick,minimum width=5mm] (B10) at (0.75,0) {\small\textcolor{cbOrange}{\faIcon{bomb}}};
        \node[,minimum width=4mm] (B11) at (1.5,0) {};
        \node[label={[reload,label distance=2pt]90:{\small$8N+6$}},thick,minimum width=5mm] (B12) at (2.25,0) {\textcolor{cbOrange}{\small\faIcon{bomb}}};
        \draw (B10) -- (B11) -- (B12);

        \foreach[count=\i] \x in {0.5,0.8,1.1,1.9,2.2,2.5} {
            \node[minimum width=2.5mm] (u\i) at (\x,-1) {};
            \draw (B11) -- (u\i);
        }
        \node[draw=none] at (1.51,-1) {\tiny$\cdots$};

        \draw[thick,decoration={brace,mirror,raise=0.25cm},decorate] (u1.west) -- (u6.east);
        \node[draw=none] at (1.5,-1.55) {\small$8N+5$};

        \node[label={[reload, label distance=2pt]90:{\small$1$}},thick,minimum width=5mm] (B20) at (4.25,0) {\small\textcolor{cbOrange}{\faIcon{bomb}}};
        \node[,minimum width=4mm] (B21) at (5,0) {};
        \node[label={[reload,label distance=2pt]90:{\small$4N$}},thick,minimum width=5mm] (B22) at (5.75,0) {\textcolor{cbOrange}{\small\faIcon{bomb}}};
        \draw (B20) -- (B21) -- (B22);

        \foreach[count=\i] \x in {4.3,4.6,5.3,5.6} {
            \node[minimum width=2.5mm] (u\i) at (\x,-1) {};
            \draw (B21) -- (u\i);
        }
        \node[draw=none] at (4.99,-1) {\tiny$\cdots$};

        \draw[thick,decoration={brace,mirror,raise=0.25cm},decorate] (u1.west) -- (u4.east);
        \node[draw=none] at (5,-1.55) {\small$4N-1$};

        \node[label={[reload, label distance=1pt]270:{\small$1$}},minimum width=5mm,thick] (D0) at
          (4.5,-2.5) {\textcolor{cbOrange}{\faIcon{bomb}}};
        \node[,minimum width=5mm] (D1) at (5.25,-2.5) {};
        \node[label={[reload, label distance=1pt]270:{\small$1$}},thick,minimum width=5mm] (D2) at (6,-2.5) {\textcolor{cbOrange}{\faIcon{bomb}}};
        \node[minimum width=2.5mm] (D3) at (6.75,-2.5) {};
        \node[minimum width=2.5mm,draw=none] (D4) at (7.5,-2.5) {\small$\cdots$};
        \node[minimum width=2.5mm] (D5) at (8.25,-2.5) {};
        \node[thick,minimum width=5mm] (D6) at (9,-2.5) {};
        \draw (D0) -- (D1) -- (D2) -- (D3) -- (D4) -- (D5) -- (D6);

        \draw[thick,decoration={brace,mirror,raise=0.225cm},decorate] (D3.west) -- (D5.east);
        \node[draw=none] at (7.55,-3.03) {\small$N^5$};
        
        \coordinate (Joint) at (10.5,0);
        \node[label={[reload,label distance=2pt]90:{\small$3N$}}] (v1) at (11.5,0) {\textcolor{cbOrange}{\faIcon{bomb}}};

        \node[ultra thick,minimum width=7mm] (t) at (13,0) {$t$};

        \node[draw,right of=B22,node distance=2cm] (S2) {\small\textcolor{cbOrange}{\faIcon{bomb}}};
        \node[draw,above of=S2,label={[reload,label distance=1pt]90:{\small$3$}}] (S1) {\small\textcolor{cbOrange}{\faIcon{bomb}}};
        
        \node[draw=none,below of=S2,yshift=5pt] (Sdummy) {$\vdots$};
        \node[draw,below of=Sdummy] (S3) {\textcolor{cbOrange}{\faIcon{bomb}}};

        \draw (B22.east) to[out=0,in=180,distance=1.5cm] (S1);
        \draw (B22.east) to[out=0,in=180,distance=1.5cm] (S2);
        \draw (B22.east) to[out=0,in=180,distance=1.5cm] (S3);
        
        \node[draw,circle,right of=S1,node distance=1.5cm,minimum width=4mm] (u1) {};
        \node[draw,circle,below of=u1,minimum width=4mm] (u2) {};
        \node[draw=none,below of=u2,yshift=5pt] (udummy) {$\vdots$};
        \node[draw,circle,below of=udummy,minimum width=4mm] (u3) {};
        \draw (Joint) edge (v1.west);
        \draw (S1) edge (u1) edge (u2) edge (udummy);
        \draw (S2) edge (u1) edge (u2) edge (u3);
        \draw (S3) edge (u3) edge (udummy) edge (u1);

        \draw (u1.east) to[out=0,in=180,distance=1cm] (Joint);
        \draw (u2.east) to[out=0,in=180,distance=1cm] (Joint);
        \draw (u3.east) to[out=0,in=180,distance=1cm] (Joint);

        \draw (s) -- (B10);
        \draw (B12.east) edge (B20.west);
        \draw (B12.east) to[out=0,in=180,distance=1.5cm] (D0.west);
        \draw (D6.east) to[out=0,in=180,distance=1.8cm] (v1.west);
        \draw (v1) -- (t);
    \end{tikzpicture}
    \caption{Illustration of the construction
    from~\Cref{thm:traps-limited-hardness}. The dark labels by the traps
    (\textcolor{cbOrange}{\faIcon{bomb}}) depict the respective reload times.}
    \label{fig:NPc:construction}
\end{figure*}
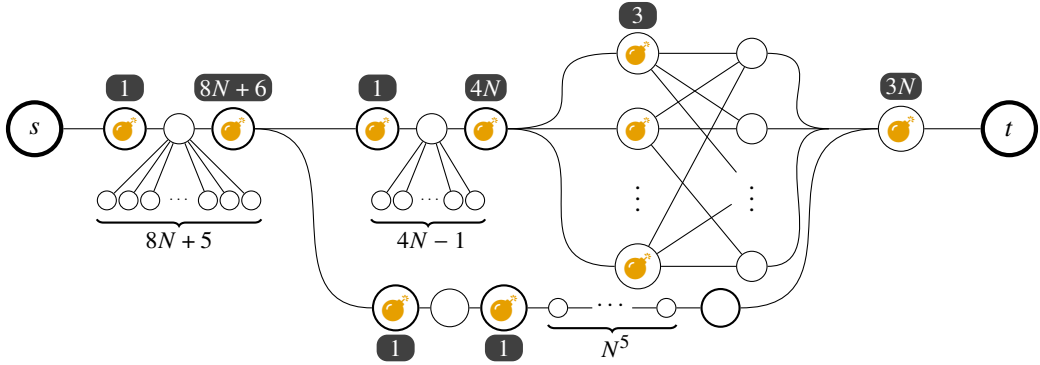

We now turn to exploring the limits of efficient computability by showing the
computational intractability of \basicProb{} by proving 
\NP-hardness.
In general, our goal is to show the hardness for instances that
are as constrained as possible, thereby precisely identifying the borders of
tractability.%

\begin{restatable}{theorem}{thmtrapslimitedharndess}\label{thm:traps-limited-hardness}
  \deferredproofmark{thm:traps-limited-hardness}%
  \basicProb{} is \NPc, even if each asset crosses at most six traps.
\end{restatable}
\begin{proof}
    We reduce from the \NPc~\citep{Gonzalez1985} \probName{Restricted Exact Cover by~$3$-Sets} problem. The input of this problem consists of a universe~$\mathcal{U}=\{u_1,\ldots,u_{3N}\}$ and a family of subsets~$\mathcal{S} = (S_1,\ldots,S_{3N})$ such that~$\forall j\in[3N]\colon S_j\in\binom{\mathcal{U}}{3}$ and each element~$u\in\mathcal{U}$ appears in exactly~$3$ subsets~$S\in\mathcal{S}$. The goal is to decide whether there is a collection~$C\subseteq \mathcal{S}$ of size~$N$ such that~$\bigcup_{S\in C} S = \mathcal{U}$.

    Let~$\mathcal{I}$ be an instance of the \probName{RX3C} problem. We
    construct an equivalent instance~$\mathcal{J}$ of~\basicProb{} using several
    gadgets. We first focus on them one by one, and then describe how they work
    together; \Cref{fig:NPc:construction} illustrates the whole construction.

    The core ingredient of our construction is a \emph{batch gadget}~$B(x)$,
    which ensures that, assuming an optimal solution in terms of the number of
    lost assets and enough assets on the input,~$x$ assets leave the gadget in
    consecutive rounds; i.e., in one batch. Formally, the batch gadget is a star
    with~$x+1$ leaves~$v_0,\ldots,v_{x}$ and the center~$c$. The leaf~$v_0$ is
    called an \emph{input port} and the leaf~$v_{x}$ is an \emph{output port}.
    The gadget contains two traps. The first is placed on the input port with reload time one. The other is placed on the output port with
    reload time equal to the parameter~$x$. The crucial idea of the gadget is
    that we need to store~$x-1$ assets on the leaves~$v_1,\ldots,v_{x-1}$ and
    only after all the leaves are occupied, the trap on the output port is
    deactivated. Observe that due to the trap at the input port, we lose at
    least one asset for every asset that ever visits the center~$c$, and we lose
    at least one asset due to the trap on the output port. Consequently, to
    pass~$x$ assets through the gadget, we lose at least~$x+2$ assets.

    Now, we formally describe the construction (see \Cref{fig:NPc:construction}
    for an illustration). We have three batch gadgets besides the initial
    vertex~$s$ and terminal vertex~$t$. The first batch gadget~$B(8N+6)$ is
    connected through its input port to the initial vertex. The second batch
    gadget~$B(4N)$ is connected to the first batch gadget and is followed
    by~$3N$ \emph{set vertices}~$s_1,\ldots,s_{3N}$, each protected by a trap
    with reload time~$3$. Next, we have~$3N$ \emph{element
    vertices}~$u_1,\ldots,u_{3N}$, which are unprotected. There is an
    edge~$\{s_j,u_i\}$ if and only if~$u_i\in S_j$ and all element vertices are
    connected with a \emph{guard vertex}~$g$. Moreover, there is a trap
    protecting the guard vertex with reload time~$3N$. The third batch
    gadget~$B(1)$ is also connected to the first gadget and is followed by a
    \emph{slowdown path} of length~$N^5$, which again ends in the guard vertex.
    The guard vertex is connected to the terminal vertex. Finally, we
    set~$\agents=\{a_1,\ldots,a_{16N+14}\}$ and~$\goal = 3N$.

The whole construction guarantees that the assets split between the set
vertices, and the slowdown path. In an optimal solution, at some point, all
element vertices are occupied by assets waiting for a couple of rounds for an
asset using the slowdown path. This asset eventually deactivates the trap on the
guard vertex and allows all assets waiting on the element vertices to pass
through~$g$ to the terminal vertex~$\tV$.
\end{proof}

Furthermore, by reducing from the $\NPh$ \threePart problem~\citep{GareyJ1975}, we rule out
polynomial-time computability %
for trees with maximum degree~$3$.

\begin{restatable}{theorem}{thmtreesnphardness}\label{thm:np-hardness-trees}
  \deferredproofmark{thm:np-hardness-trees}%
  \basicProb is \NPc, even if the topology~$G$ is a tree with maximum degree
  $3$.
\end{restatable}

In light of the established hardness, it is natural to consider approximations.
Let~\optBasicProb{} be a natural optimization variant of our problem, in which
we seek the maximum number of assets that survive. Importantly, the proof
of~\Cref{thm:noapx} excludes a constant-factor approximation algorithm
for~\optBasicProb{} (under standard assumptions).\footnote{\optBasicProb{} is
not only outside of~\APX{} but also \APXh{}, as we show
in~\Cref{app:apx-hardness}.}

\begin{restatable}{theorem}{thmnoapx}\label{thm:noapx}
  \deferredproofmark{thm:noapx}%
  There is no polynomial-time algorithm yielding a multiplicative
  constant-factor approximation to~\optBasicProb{}, even if topology~$G$ is a
  tree with maximum degree~$3$.
\end{restatable}

\section{Conclusions}

We introduced a new planning model for adversarial environments in which agents
may be eliminated while traversing hostile terrain. We established that the
problem is in \NP despite the exponentially large plan space, proved \NPhness
even on trees of constant degree, and identified several tractable cases. Our
framework captures applications ranging from humanitarian aid transportation in
conflict zones to coordinated drone swarm navigation.

Our results open multiple promising avenues for future research. In this initial
study, we assumed complete knowledge of the environment. Since such full
observability is rare in real-world scenarios, a natural extension is to
incorporate
\emph{uncertainty}~\citep{NikolovaBK2006,WagnerC2017,ShoferSS2023} in trap
locations, reload times, or graph topology.

Additionally, our hardness results motivate the design and analysis of efficient
\emph{heuristics and approximations}. While constant-factor approximations may
be limited to special cases due to our inapproximability result, heuristics,
albeit without formal optimality guarantees,  could yield high-quality solutions
in many practical instances.

Finally, it is interesting to consider \emph{multiple start or target locations}
to broaden the model's applicability to scenarios such as distributed supply
networks or multi-objective delivery missions. 

\section*{Acknowledgments}
Andrzej Kaczmarczyk and Haifeng Xu were partially supported by the ONR Award
N00014-23-1-2802. Šimon Schierreich was
partially funded by the European Research Council (ERC) under the European
Union’s Horizon 2020 research and innovation programme (grant agreement No
101002854). Additionally, Andrzej Kaczmarczyk and Šimon Schierreich were
partially funded by the European Union under
the project Robotics and advanced industrial production (reg. no.
CZ.02.01.01/00/22\_008/0004590). Most of this work was done while Šimon
Schierreich was a Fulbright-Masaryk Fellow at Pennsylvania State
University.

\begin{center}
\vspace{0.25cm}
\includegraphics[width=3.5cm]{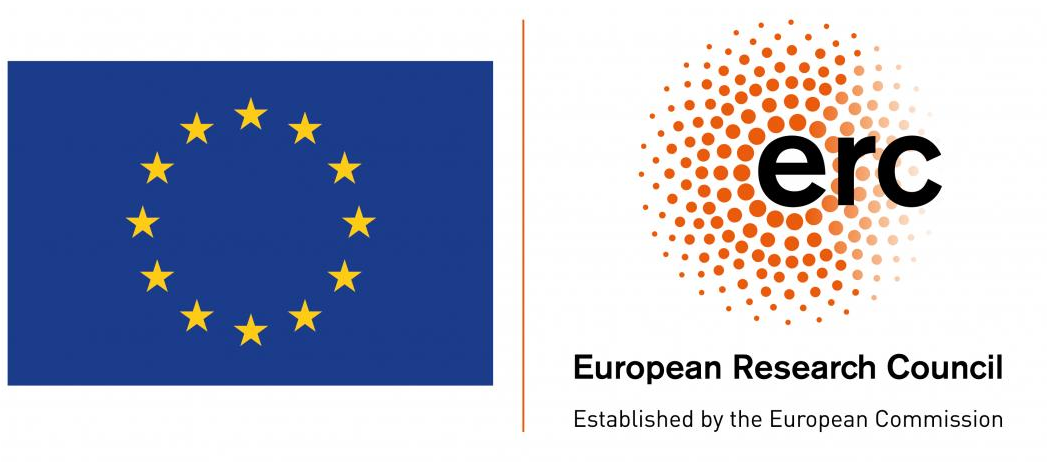}
\end{center}

\bibliographystyle{abbrvnat}
\bibliography{biblio}

\clearpage
\appendix
\section*{Appendix}

\section{Omitted Proofs}\label{app:proofs}
\restatetheorem{\clmPolyPathNoReturns*}
\begin{claimproof}
  \backtotheorem{clm:path:poly:noReturns}%
  Let $a$ be an asset such that there exists a time step $j$ such that
  $\plan(a,j) v_i$ and $\plan(a,j+1) = v_{i-1}$ for some $i\in[\numVertices]$;
  we call such an asset \emph{returning}. Notice that if $i=1$, then the plan
  $\plan$ is invalid, s by property 1, the assets cannot return to $s$. Hence,
  $i \geq 2$. Moreover, assume that $a$ is a returning assets which returns in
  the smallest timestep. If here are multiple such assets, let $a$ be the
  `leftmost', i.e., the one located n location $v_i$ closest to $s$.
    
  First, assume that there is $\ell$ such that $\plan(a,j+\ell) = \plan(a,j)$.
  That is, the asset $a$, after $\ell$ steps, returns to location $v_i$. If
  $v_i\not\in\traps$, then we set $\plan(a,j') = v_i$ for every
  $j'\in[j,\ell]$; that is, we make agent $a$ wait on its location. Note that
  since the topology is a path, there cannot be an agent that is newly blocked
  by $a$, as there is no way how to circumvent it. Moreover, since
  $v_i\not\in\traps$, then $a$ can wait on it's location without being
  destroyed.
\end{claimproof}

\restatetheorem{\clmPolyPathMove*}
\begin{claimproof}
  \backtotheorem{clm:path:poly:move}%
  Observe that the move $\plan'(a,j+1) = v_{i+1}$ does not change the timing of
  any trap, as $v_{i+1}$ is either a regular vertex or an inactive trap (in
  round $j+1$). Therefore, as the plan for each other asset remains the same, this
  change does not affect whether they succeed or not. Moreover, asset~$a$ moves
  to a neighboring empty vertex closer to $\tV$ and continues its movement in
  round $j+\ell$, where $j+\ell$ is the round in which $a$ leaves $v_{i+1}$ in
  plan $\plan$. Hence, if $a$ is successful in $\plan$, it also succeeds in
  $\plan'$.
\end{claimproof}

\restatetheorem{\clmPolyPathWaitSacrifice*}
\begin{claimproof}
  \backtotheorem{clm:path:poly:waitSacrifice}%
  We prove the claim by induction. Let $S$ be the last segment and $r_l$ and
  $r_p$ be the left and right trap of this segment.
  By~\Cref{clm:path:poly:move}, whenever assets can move to an empty vertex
  closer to $\tV$, then they do so. In particular, all assets positioned on a
  vertex after $r_p$ in round~$j$ will move closer to~$\tV$ in round $j+1$, and
  ultimately they will reach $\tV$ as soon as possible. This implies that there
  is no asset waiting on the vertex after $r_p$ (except in the case where this
  vertex is $\tV$, but in this case, we can simply subdivide the edge
  $\{r_p,\tV\}$ without changing the solution).

  Moreover, let $\plan$ be a plan in which $a$ deactivates trap~$r_p$ sooner
  than when the whole segment is full (slightly abusing the terminology, by
  full we mean that there are $\ell'=\min\{|S|,c(r_p),|\{b\in
  A\setminus{a}\mid\plan(b,j)=v_{\ell''}\land \ell''\leq\ell\}|\}$ assets
  following~$a$). We let $a$ wait once the segment saturates (this necessarily
  happens due to~\Cref{clm:path:poly:move}). Only then asset~$a$
  deactivates the trap and at least $\ell'$ assets succeed. Observe that there
  is no agent waiting on the first vertex of $S$ (the one in $N(r_\ell)\cap
  S$), so we can use the same procedure for all segments preceding $S$,
  finishing the proof.
\end{claimproof}

\restatetheorem{\propdisjointpaths*}
\begin{proof}
  \backtotheorem{prop:disjoint-paths}%
  \newcommand{\DP}{\operatorname{DP}}
    Our algorithm is a dynamic programming algorithm over the disjoint paths we obtain by deleting~$\sV$ and~$\tV$ from the topology~$G$, which uses the algorithm for a path as a sub-procedure. We formalize our approach using~\Cref{alg:disjointPaths}. Intuitively, the core of the algorithm is a dynamic programming table $\DP[i,B]$, which stores the maximum number of successful assets on the sub-graph $G[\bigcup_{j=1}^i V(P_j) \cup \{\sV,\tV\}]$ if we have exactly $B$ assets available. Once the table is correctly computed, we simply check whether $\DP[\ell,|\agents|] \geq \goal$ and return an appropriate response.
    
    \begin{algorithm}[tb]
        \caption{A dynamic programming algorithm finding an optimal solution for \basicProb if $G\setminus\{\sV,\tV\}$ is a disjoint union of paths.}
        \label{alg:disjointPaths}
        \begin{flushleft}
            \textbf{Input}: A set of assets~$\agents$, a topology~$G$ with an initial and target location~$\sV$ and~$\tV$ such that~$G\setminus\{\sV,\tV\}$ is a disjoint union of paths~$P_1,\ldots,P_\ell$, a set of traps~$\traps$, a reload time function~$\reloadFn$.\\
            \textbf{Output}: The maximum number of successful assets.
        \end{flushleft}
        \begin{algorithmic}[1] %
            \State \Return \Call{SolveRec}{$\ell,|\agents|$}
        
            \Function{SolveRec}{$i,B$}
                \If{$i=1$}
                    \State\Return \Call{SolvePath}{$P_1,B$}\label{alg:disjointPaths:basicStep}
                \EndIf
                \If{$\DP[i,B] = \texttt{undefined}$}
                    \State~$r \gets -\infty$
                    \For{$\forall b\in[B]_0$}
                        \State~$x \gets \Call{SolvePath}{P_i,b}$\label{alg:disjointPaths:locallyOptimal}
                        \State~$y \gets \Call{SolveRec}{i-1,B-b}$\label{alg:disjointPaths:inductionHypothesis}
                        \State~$r \gets \max\{r,x + y\}$
                    \EndFor
                    \State~$\DP[i,B] \gets r$
                \EndIf
                \State \Return~$\DP[i,B]$
            \EndFunction
        \end{algorithmic}
    \end{algorithm}

    We show correctness of the computation by induction over $i$. First, let
    $i=1$. Then, $G[V(P_i) \cup \{\sV,\tV\}]$ is a simple path. Therefore, we
    can use the algorithm of \Cref{thm:path:poly} to find an optimal plan
    solving such a sub-graph, which is exactly what the basic step of
    \Cref{alg:disjointPaths}; see line~\Cref{alg:disjointPaths:basicStep}. Now,
    let $i > 1$ and assume that $\DP[j,\beta]$ is computed correctly for all
    $j\in[i]$ and $\beta\in [|\agents|]_0$. Recall that in any valid plan, once
    an asset leaves $\sV$, it can never return to it. Hence, if we decide to
    send some assets through $P_i$, these assets never visits any vertex of
    paths $P_1,\ldots,P_{i-1}$, as $G\setminus\{\sV,\tV\}$ is a collection of
    disjoint paths. Hence, if we decide to use $P_i$ for some $b\in[B]_0$ assets,
    the solution for $P_i$ is independent of the solution on
    $P_{1},\ldots,P_{i-1}$. Moreover, on graph $G[V(P_i)\cup\{\sV,\tV\}]$, an
    optimal solution follows the algorithm of \Cref{thm:path:poly}, which is
    also what our algorithm does on
    line~\Cref{alg:disjointPaths:locallyOptimal}. Moreover, by induction
    hypothesis, if we combine this partial plan with the result return on
    line~\Cref{alg:disjointPaths:inductionHypothesis}, we clearly obtain an
    optimal solution for this $b$. Since we try all possible $b\in[B]_0$, the
    algorithm is indeed optimal.
    
    The size of the dynamic programming table~$\DP$ is~$\Oh{|V(G)|\cdot\numAgents}$, and each cell can be computed in~$\Oh{\numAgents\cdot T_{\operatorname{path}}}$ time, where~$T_{\operatorname{path}}$ is the running time of the algorithm returning the maximum number of successful assets if~$G$ is a path. By \Cref{thm:path:poly},~$T_{\operatorname{path}}\in\operatorname{poly}(|\mathcal{I}|)$, so our approach also runs in polynomial time.
\end{proof}

\restatetheorem{\propdisjoinstructure*}
\begin{proof}
  \backtotheorem{prop:disjoint-structure}%
  We transform instance~$\mathcal{I}$ to an equivalent new
  instance~$\mathcal{J}$ whose topology is collection of disjoint paths
  connecting the start and target locations. Then,
  invoking~\Cref{prop:disjoint-paths}, we solve~$\mathcal{I}$.

  Let $G'=(V',E')$ be a condensed topology of~$G=(V,E)$. We obtain a new
  instance~$\mathcal{J}= (H, \sV,\tV,\traps,\reloadFn,\agents,\goal)$
  of~\basicProb{}, in which the graph $H$ is constructed as follows. We start by
  taking the copy of~$G'$ as~$H$. Then, for each vertex~$v' \in V' \setminus
  (\{\sV, \tV\} \cup R)$, we substitute this vertex in~$H$ with a path of length
  equal to the number of vertices of the connected component represented by~$v'$
  in a way that the two neighbors of the substituted~$v'$ are connected to the
  two ends of the path.

  The equivalence of~$\mathcal{I}$ and~$\mathcal{J}$, follows mainly from the
  fact that the ideas and proofs
  from~\Cref{clm:path:poly:move,clm:path:poly:waitSacrifice} apply
  to~$\mathcal{I}$ almost without changes. The only difference is that
  for~$\mathcal{I}$, we do not order vertices according to the path between the
  previous and the next trap, but according to the distance from the next trap
  (breaking possible ties arbitrarily). Having that and the fact that each trap
  has degree two in~$G$, allows us to easily transform each solution
  of~$\mathcal{I}$ to~$\mathcal{J}$ and vice versa, which proves the
  proposition. 
\end{proof}

\restatetheorem{\propstarcase*}
\begin{proof}
  \backtotheorem{prop:star-case}
  Let~$\mathcal{I} = (G,\sV,\tV,\traps,\reloadFn,\agents,\goal)$ be a \basicProb
  instance, where~$G$ is a star. We give a plan that implies that \basicProb is
  solvable on~$G$ in~$\Oh{1}$ time. The plan is to move the assets one after
  another with no delay from~$\sV$ to~$\tV$, call this plan~$\plan$. We show
  that~$\plan$ is optimal and that it is easy to compute the number~$k'$ of
  assets that succeed. Before we start, note that the unique path from~$\sV$
  to~$\tV$ must pass through a trap vertex, as per our assumptions on~$G$ (from
  the beginning of the section on our algorithms). 
  
  Let~$r$ be the trap vertex with the reload time~$c \coloneq \reloadFn(r)$
  (recall that~$\reloadFn(r) > 0$) on the~$\sV$-$\tV$ path and let us decompose
  the number~$\numAgents > 0$ of assets into~$\numAgents \coloneq d(c+1)+q$
  with~$d \geq 0$ and~$0\leq q<c +1$. Under~$\plan$, every eliminated asset
  allows at most~$\reloadFn(v)$ assets to pass. So, for every~$c+1$~assets there
  are~$c$~assets that pass. Thus, if~$q=0$, we lose exactly~$d =
  \nicefrac{\numAgents}{c+1}$~assets and
  get~$k'=\numAgents-\nicefrac{\numAgents}{c+1}=\nicefrac{c\numAgents}{c+1}$
  successful ones. However, if~$1\leq q<c+1$, then right after the
  above-mentioned~$d \cdot (c+1)$~assets move past the trap (or are eliminated
  thereon), the trap becomes active. So, we lose one more asset so that the
  remaining~$q-1$ (possibly~$0$) assets can pass. Since~$q<(c+1)$
  implies~$q-1<c$, we conclude that~$k'=\numAgents-\lfloor
  \nicefrac{\numAgents}{c+1} \rfloor-1=\numAgents-\lceil
  \nicefrac{\numAgents}{c+1} \rceil$ succeed. Hence, altogether we
  get~$k'=\numAgents-\lceil \nicefrac{\numAgents}{c+1} \rceil$~successful
  assets, regardless of whether~$q=0$ or~$q\neq0$.

  To see the optimality of~$\plan$, consider some~$\numAgents$~consecutive
  rounds such that in the first round~$r$ is active. By definition, there must
  be at least one elimination for every block~$c$ of rounds when~$r$ is
  inactive. Since the sequence of rounds starts with~$r$ being active, then~$r$
  can be inactive for at least~$\lceil \nicefrac{\numAgents}{c + 1} \rceil$
  rounds in the whole sequence. Since~$\plan$ moves~$\numAgents$ assets
  through~$r$ consecutively and~$r$ is initially active, we obtain that out
  of~$\numAgents$~rounds~$r$ is active through at most~$k' = \numAgents - \lceil
  \nicefrac{\numAgents}{c + 1} \rceil$~rounds. So, we conclude that~$\plan$ is
  optimal.

  Finally, to solve instance~$\mathcal{I}$ it is enough to verify
  in~$\Oh{1}$~time whether~$\goal \leq \numAgents - \lceil
  \nicefrac{\numAgents}{c + 1} \rceil$ indicating a \Yes-instance.
\end{proof}

\restatetheorem{\thmtrapslimitedharndess*}
\begin{proof}[Proof of correctness]
  \backtotheorem{thm:traps-limited-hardness}%
    For correctness, assume that~$\mathcal{I}$ is a \Yes-instance and~$X\subseteq \mathcal{S}$ is an exact cover of~$\mathcal{U}$. Starting with the first round, in consecutive rounds, the assets move from the initial vertex to the input part, and, eventually, continue to some leaf of the first batch gadget. Every odd asset is destroyed by the trap protecting the input port of~$B(8N+6)$. Once all leaves of~$B(8N+6)$ are full, the asset~$a_{16N+12}$ waits for one round at the center of the first batch gadget. In the other round, asset~$a_{16N+12}$ inactivates the trap on the output port and is replaced by asset~$a_{16N+14}$ in the center of~$B(8N+6)$. Now, all assets leave, in lexicographic order, the first batch gadget in consecutive rounds, which takes~$8N+6$ rounds and, therefore, none of them is destroyed by the trap on the output port of~$B(8N+6)$. First four assets leaving~$B(8N+6)$, namely~$a_{16N+14}$,~$a_{2}$,~$a_4$, and~$a_6$, crosses by the same strategy the third batch gadget~$B(1)$; consequently, only the asset~$a_6$ passes the output port of~$B(1)$ without being destroyed and continues, without waiting, through the slowdown path towards the guard vertex. All the remaining assets use a path via the second batch gadget, again using the same strategy.~$4N+1$ of them are destroyed on the input port of~$B(4N)$, and one is destroyed at the output port. Hence, exactly~$8N+2 - (4N+1) - 1 = 4N$ assets successfully pass~$B(4N)$. Now, let~$X = \{S_{i_1},\ldots,S_{i_N}\}$. Then, the first four assets leaving~$B(4N)$ go through set vertex~$s_{i_1}$, the second four assets through set vertex~$s_{i_2}$, etc. Obviously, the first asset in each group is destroyed, but all three remaining assets can pass through the set vertex to an empty element vertex. Note that each set vertex is adjacent to exactly three element vertices, and since~$X$ is an exact cover, there is always an empty element vertex for each passing asset. Once all element vertices are occupied, these assets are waiting till the asset~$a_6$ inactivates the trap on the guard vertex and, after that, all of them reach the terminal vertex. As was shown, exactly~$3N$ assets succeeded in this plan, so~$\mathcal{J}$ is indeed a \Yes-instance.

    In the opposite direction, let~$\mathcal{J}$ be a \Yes-instance and~$\plan$ be a plan such that at least~$3N$ assets succeed. We split the proof into several claims. First, we show that the ``upper path'' is used at least once.

\begin{claim}\label{clm:path:NPh:topPath}
    The trap on the output port of~$B(4N)$ is deactivated at least once.
\end{claim}
\begin{claimproof}
        For the sake of contradiction, assume that all assets go through the slowdown path. Then, each of them needs to cross~$B(1)$. By the properties of the batch gadget, to pass~$1$ asset, at least~$1+2 = 3$ assets are destroyed. As at most~$8N+6$ assets can cross~$B(8N+6)$, we obtain that at most~$\lfloor (8N+6)\cdot \frac{1}{4}\rfloor = \lfloor 2N + \frac{3}{2} \rfloor = 2N + 1$ assets enters the slowdown path, which contradicts that at least~$3N$ assets succeed. 
    \end{claimproof}

    Similarly, the ``bottom part'' is used at least once. In fact, we show a stronger property that it is actually used by exactly one asset.

\begin{claim}\label{clm:path:NPh:slowdownPath}
    The trap on the output port of~$B(1)$ is deactivated exactly once.
\end{claim}
\begin{claimproof}
        First, assume that it is not deactivated at all. For the sake of contradiction, assume that all assets go through the~$B(4N)$ gadget. Then, as~$8N+6$ assets leave~$B(8N+6)$, at least two assets are destroyed on the output port of~$B(4N)$, and to pass them over the input port of~$B(4N)$, two more assets are destroyed. After the first deactivation of the trap in the output port, at most~$4N$ assets pass through the output port, and the same number of assets are destroyed to open their path through the input port. Consequently, after the second inactivation of the output port, at most~$(8N+6-4-4N-4N)/2 = 2/2 = 1$ asset can pass. That is, there are at most~$4N+1$ assets surviving after the output port of~$B(4N)$. Each of them needs to pass through the traps on the set-vertices. These traps have a reload time~$3$, meaning that at most~$\lfloor 3\cdot4N/4 \rfloor = 3N$ can pass through them. Finally, at least one of them is destroyed on the guard-vertex, meaning that at most~$3N-1$ assets successfully reach~$t$, which is a contradiction with~$\plan$ being a solution.
    \end{claimproof}

    Consequently, we see that the trap on the output port of $B(4N)$ is also deactivated exactly once. Next, we prove that exactly $N$ (disjoint) traps placed on the set vertices are deactivated. By \Cref{clm:path:NPh:slowdownPath}, exactly four assets enter the batch gadget at the beginning of the slow-down path. Hence, exactly $8N+2$ assets enter $B(4N)$ and exactly $4N$ of them leave this gadget in consecutive rounds by \Cref{clm:path:NPh:topPath}. Since exactly one asset follows the slowdown path, we need to pass exactly $3N$ agents to the element vertices. The traps on the set vertices have reload time $3$ and therefore we lose at least $4N/4 = N$ agents by deactivating them. That is, necessarily exactly $N$ of these traps are deactivated. Moreover, each set vertex is adjacent to exactly three element vertices, and there is no way for assets to move from one element vertex to another without deactivating a trap. Hence, if a trap on a set vertex is deactivated twice, then their neighboring element vertices must be empty. This can only happen if the guard trap is activated. However, the asset $a$ using the slowdown part can reach it at the beginning of the round $N^5$, while the last asset leaves $B(4N)$ in time at the latest $12N+3$, which means that $a$ reaches the guard trap once it is active again, and we have only $3N-1$ successful agents, which is in contradiction to the fact that $\plan$ is a solution. Hence, all assets must pass the guard trap in consecutive rounds, which means that we had to use exactly $N$ distinct set vertices. The set vertices used correspond exactly to the solution of $\mathcal{I}$. This finishes the proof.
\end{proof}

\restatetheorem{\thmtreesnphardness*}
\begin{proof}
  \backtotheorem{thm:np-hardness-trees}%
  \begin{figure}
  \centering%
    \begin{tikzpicture}[node distance=1.3cm,minimum size=0.5cm,
      font=\footnotesize,
      every node/.style={draw,minimum width=5mm,inner sep=1pt}
      ,label distance=-2mm]
    \node[very thick, circle, draw] (s) {$s$};

    \coordinate (LSepB) at (-1.75,-1.1);
    \coordinate (LSepMid) at (4,-4);
    \draw[ultra thick, dashed] (LSepB) .. controls (4.25, -1.1) .. (LSepMid);

    \node[draw,circle] (main) at (-1,-1.7) {{\textcolor{cbOrange}{\faIcon{bomb}}}};
    \node[reload,draw=none, below = 1pt of main] {$s(\threePartDummy_i)+1$};
    \node[draw=none, above of=main,node distance=0.4cm] {\textit{Element Gadget:}};

    \node[draw,circle, right of=main] (one) {};
    \node[draw=none, right of=one] (dots) {$\dots$};
    \node[draw,circle,right of=dots] (last){};
    \draw (main) -- (one) -- (dots) -- (last);

    \draw [decorate,decoration={brace,amplitude=6pt,mirror}] 
    ($(one.south)-(0,0pt)$) -- ($(last.south)-(0,0pt)$)
    node[draw=none,fill=none,midway,below=7pt] {$s(\threePartDummy_i)+1$};

    \node[circle, draw] (x) at (-1,-3.2) {{\textcolor{cbOrange}{\faIcon{bomb}}}};
    \node[reload,draw=none,below = 1pt of x] {$c(v)$};
    \node[draw=none,above of=x, node distance=0.5cm, xshift=0.2cm] {\textit{Destruction Gadget:}};
    \node [circle, draw, right of=x] (x1)  {{\textcolor{cbOrange}{\faIcon{bomb}}}};
    \node[reload,draw=none,below = 1pt of x1,xshift=5pt] {$c(v)-1$};

    \node [circle,draw=none, right of=x1] (dots) {$\dots$};

    \node [circle,draw, right of=dots] (end) {{\textcolor{cbOrange}{\faIcon{bomb}}}};
    \node[reload,draw=none, below = 1pt of end] {$c(v)-\ell+1$};

    \draw (x) -- (x1) -- (dots) -- (end);
    
    \node[right of=s,circle, draw] (d1) {{\textcolor{cbOrange}{\faIcon{bomb}}}};
    \node[reload,draw=none,below left= 8pt and 17pt of d1,anchor=north] (clock) {$100n^5T^5+100nT+(3T+19+3n)$};
    \draw[densely dotted] (d1.south) -- (clock.north);

    \node[draw=none,right of=d1,circle] (ddots) {$\dots$};
    \node[right of=ddots,circle, draw] (dend) {{\textcolor{cbOrange}{\faIcon{bomb}}}};
    \node[reload,draw=none,below = 6pt of dend,anchor=north] (clock) {$3T+20+3n$};
    \draw[densely dotted] (dend.south) -- (clock.north);

    \node[right = 0.45cm of dend,circle, draw] (mid1) {{\textcolor{cbOrange}{\faIcon{bomb}}}};
    
    \node[reload,draw=none,above of=mid1,node distance=1.7cm] (clock) {$(n+3)T+3n+20+n(T+7)-(T+6)\lceil\frac{n(T+7)}{2T+13}\rceil$};

    \draw[decorate,decoration={brace,amplitude=10pt}] (d1.north) -- (dend.north)
      node[midway,xshift=-1cm,draw=none] (brace) {};
    \node[draw=none,above of=brace,xshift=1cm,node distance=0.5cm] {Limiting Infrastructure};

    \node[below of=mid1, draw,circle] (st1){{\textcolor{cbOrange}{\faIcon{bomb}}}};
    \node[reload,draw=none,below of=st1, node distance=0.5cm] {$s(\threePartDummy_i)+1$};

    \node[below = 15pt of st1,draw,circle] (st2){};
    \node[below right= 15pt and 35pt of st2, draw,circle] (st4x) {};
    \draw (st2) to node[midway,circle,draw=none,rotate=-26,fill=white]{\dots}
      (st4x);
    \draw (mid1) -- (st1) -- (st2) ;

    \draw[decorate,decoration={brace,mirror,amplitude=5pt,raise=7pt,aspect=0.3}]
      (st2.center) -- (st4x.center)
      node[pos=0.3,draw=none,below=11pt,sloped] (brace) {$s(\threePartDummy_i)+1$};
    
    \node[draw=none,right = .5cm of mid1,circle] (mid2) {\dots};
    \node[below of=mid2, draw,circle] (st1){{\textcolor{cbOrange}{\faIcon{bomb}}}};

    \node[below = 15pt of st1,draw,circle] (st2){};
    \node[below right= 15pt and 35pt of st2, draw,circle] (st4) {};
    \draw (st2) to node[midway,circle,draw=none,rotate=-26,fill=white]{\dots} (st4);
    \draw (mid2) -- (st1) -- (st2) ;

    \node[right= .4cm of mid2,circle, draw] (mid3) {{\textcolor{cbOrange}{\faIcon{bomb}}}};
    \node[below of=mid3, draw,circle] (st1){{\textcolor{cbOrange}{\faIcon{bomb}}}};

    \node[below= 15pt of st1,draw,circle] (st2){};
    \node[below right= 15pt and 35pt of st2, draw,circle] (st4) {};
    \draw (st2) to node[midway,circle,draw=none,rotate=-26,fill=white]{\dots} (st4);
    \draw (mid3) -- (st1) -- (st2) ;
    \node[right= .5cm of mid3,circle, draw] (midEnd) {{\textcolor{cbOrange}{\faIcon{bomb}}}};

    \node[reload,draw=none,below of=st1,node distance=0.5cm] (st1c) {$s(\threePartDummy_k)+1$};
    
    \draw [decorate,decoration={brace,amplitude=6pt,mirror}]
      ($(st4x.south)-(0,0pt)$) -- ($(st4.south)-(0,0pt)$)
      node[draw=none,fill=none,midway,xshift=0cm,yshift=-0.5cm,draw=none] {$3n$
      Element Gadgets,~$y_a\neq y_b, \forall y_a,y_b\in Y$};

    \draw[densely dotted] (mid1.north) -- (clock.south);
    \draw[densely dotted] (mid2.north) -- (clock.south);
    \draw[densely dotted] (mid3.north) -- (clock.south);
    \draw[densely dotted] (midEnd.north) -- (clock.south);
    \node[draw=none,below right= 25pt and 5pt of midEnd, draw,circle] (st1){{\textcolor{cbOrange}{\faIcon{bomb}}}};
    \node[reload,draw=none, above of=st1,yshift=-0.7cm,xshift=0.5cm] (st1c) {$2T+12$};
    \draw[densely dotted] (st1c.south) -- (st1.north east);
    \node[draw,below right = 10pt and 40pt of st1,anchor=center,circle, node distance =
      25pt] (st1y){{\textcolor{cbOrange}{\faIcon{bomb}}}};
    \node[reload,draw=none,below of=st1y] (st1yc) {$2T+12$};
    \draw[densely dotted] (st1y.south) -- (st1yc.north);
    \draw (st1) to node[midway,circle,draw=none,rotate=-20,fill=white]{\dots} (st1y);

    \node[right=0.3cm of st1y,draw,circle] (st2){};
    \node[draw=none,right=0.3cm of st2,circle] (st3) {$\dots$};
    \node[right=0.3cm of st3, draw,circle] (st4) {};
    \draw (midEnd) -- (st1)
          (st1y) -- (st2) -- (st3) -- (st4);
        
    \draw[decorate,decoration={brace,amplitude=5pt,mirror}]
      ($(st1.center)+(-100:8pt)$) --
      ($(st1y.center)+(-120:8pt)$) node[midway,xshift=1cm,draw=none] (brace) {};
    \node[draw=none,anchor=north east,below left= 5pt and 5pt of brace,
      rotate=-20] {$T+6$ Traps};

    \draw[decorate,decoration={brace,amplitude=10pt}] (st4.south) -- (st2.south) node[midway,xshift=1cm,draw=none] (brace) {};
    \node[draw=none,below of=brace,xshift=-0.5cm,node distance=0.75cm] {$n(T+7)$ Empty Nodes};

    \node[right of=midEnd,circle, draw] (end1) {{\textcolor{cbOrange}{\faIcon{bomb}}}};
    \node[reload,draw=none,above left= .6cm and -3.25cm of end1] (clock) {$(n+3)T+18+n(T+7)-(T+6)\lceil\frac{n(T+7)}{2T+13}\rceil$};

    \node[draw=none,right of=end1,circle, draw] (end2) {{\textcolor{cbOrange}{\faIcon{bomb}}}};

    \node[draw=none,right of=end2,circle] (end3) {\dots};
    \node[right= .2cm of end3,circle, draw] (endEnd) {{\textcolor{cbOrange}{\faIcon{bomb}}}};
    \node[reload, above of=endEnd,yshift=-0.7cm] (endEndc) {$1$};
    
    \draw [decorate,decoration={brace,amplitude=6pt,mirror}] ($(end1.south)-(0,0pt)$) -- ($(endEnd.south)-(0,0pt)$) node[draw=none,fill=none,midway,xshift=0cm,yshift=-0.5cm] {Blocking Infrastructure};

    \node[very thick, right = .2 of endEnd,circle, draw] (goal) {$t$};

    \draw[densely dotted] (end1.north) -- (clock.south);

      \draw (s) -- (d1);
      \draw (d1) -- (ddots);
      \draw (ddots) -- (dend);
      \draw (dend) -- (mid1);
      \draw (mid1) -- (mid2);
      \draw (mid2) -- (mid3);
      \draw (mid3) -- (midEnd);
      \draw (midEnd) -- (end1);
      \draw (end1) -- (end2);
      \draw (end2) -- (end3);
      \draw (end3) -- (endEnd);
      \draw (endEnd) -- (goal);
        \end{tikzpicture}
  \caption{Construction of the reduction on trees}
  \label{fig:trees-hardness-construction}
\end{figure}
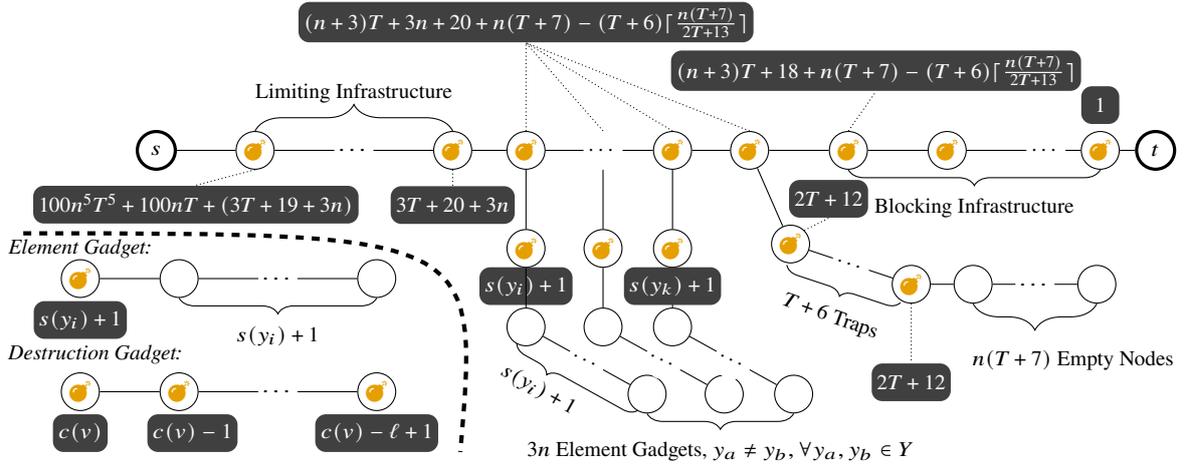%
 The problem is in \NP by~\Cref{cor:probInNP}. To show it is \NPh, we reduce
 from~\threePart{}.

 \defProblemQuestion{\threePart \cite[see p.~224]{GareyJohnson1990}}
     {Set~$Y$ of~$3n$ elements, a bound~$T \in \mathbb{Z}^+$, and a size~$s(y)
     \in \mathbb{Z}^+$ for each~$y \in Y$ such that~$T/4 < s(y) < T/2$
   and~$\sum_{y \in Y} s(y) = nT$.}%
     {Can~$Y$ be partitioned into~$n$ disjoint sets~$Y_1, Y_2, \ldots, Y_n$ of
     size exactly~$3$ such that for~$1 \leq i \leq n$,~$\sum_{y \in Y_i} s(y) =
     T$?} 
     
   Let~$\mathcal{P}$ be an instance of~\threePart{} such that~$n\geq3$ is odd
   and~$n<T$.\footnote{This problem is still~$\NPc$ since~$\binom{3n}{3}$ is
   polynomial in~$n$, and we can add a large number~$x>n$ to all~$s(y)$ and
 increase~$T$ by~$3x$} We construct an instance of
 \basicProb~$\mathcal{J}~=~(G,\sV,\tV,\traps,\reloadFn,\agents,\goal)$,
 where~$G$ is a tree of maximum degree 3, in polynomial time and show
 that~$\mathcal{P}$ is a \Yes-instance if and only if~$\mathcal{J}$ is a
 \Yes-instance. 

    For clarity we divide the reduction into 3 sections, "Construction", "Conception", which gives the intuition, and "Correctness".
    
    \textbf{Construction}.
    First, we define some gadgets, that we illustrate
    in~\Cref{fig:trees-hardness-construction}:
    \begin{enumerate}
        \item A~$\textit{destruction gadget}$ consists of~$\ell$
          traps arranged in order as a path~$v_1,v_1,\dots,v_\ell$, with the
          reload times of the vertices being~$c(v_1), c(v_1)-1, \dots,
          c(v_1)-\ell+1$, respectively.  A~$\textit{destruction gadget}$ is
          defined by a length~$\ell$ and a starting reload time~$c(v_1)$. The
          purpose of each of these gadgets in our reduction is to destroy~$\ell$
          assets while allowing~$c-\ell+1$ to pass through for every~$c+1$
          assets that move immediately one after another.
        
        \item An~\emph{element gadget} represents an element from the
          original set. For each element~$y \in Y$, the~\emph{corresponding
          element gadget} of~$y$ is a path starting with a trap with a reload
          time of~$s(y)+1$ refered to as the~\emph{guard vertex}, followed
          by~$s(y)+1$ empty vertices. The element is represented by
          storing~$s(y)+1$ assets in the empty vertices, which to obtain a
          \Yes-instance must happen (to be shown). 
        
    \end{enumerate}

    Now we formally construct the tree graph~$G$. This construction is shown in~\Cref{fig:trees-hardness-construction}.
    \begin{enumerate}
        \item Define a \textit{destruction gadget} that has starting reload time~$100n^5T^5+100nT+(3T+19+3n)$ and length~$100n^5T^5+100nT$. We refer to this \textit{destruction gadget} as the \textit{limiting infrastructure}. And call the nodes in this \textit{limiting infrastructure}~$v_i^{L}$,~$i\in\{1, \dots, 100n^5T^5+100nT\}$. We then connect~$\sV$ only to~$v_1^L$.
        \item Followed by and connected to~$v_{100n^5T^5+100nT}^L$ in the \textit{limiting infrastructure} is a path of traps of length~$3n+1$, all with reload time~$(n+3)T+3n+20+n(T+7)-(T+6)\lceil\frac{n(T+7)}{2T+13}\rceil$. Call the nodes in this path~$v_i^M$,~$i \in \{1, \dots, 3n+1\}$.
        \item For each of the nodes~$v\in v_i^M$,~$i\in\{1,\dots, 3n\}$, there is an \textit{element gadget} connected to~$v$ such that; 1) For some~$y \in Y$, we construct the \textit{corresponding element gadget} and attach it to~$v$ by connecting~$v$ and the \textit{guard vertex} and 2) no 2 \textit{element gadgets correspond} to the same~$y$
        \item There is also a subtree, call it the \textit{final subtree}, connected to~$v_{3n+1}^M$. This subtree consists of~$T+6$ traps arranged in a path, each with reload time~$2T+12$. The start of the path is connected to~$v_{3n+1}^M$. Followed by the last trap in this path is another path of~$n(T+7)$ empty nodes. This subtree is artificial and is not used to represent any elements, but will be used for the proof.        
        \item Next, we have a \textit{destruction gadget} with starting reload time~$(n+3)T+19+n(T+7)-(T+6)\lceil\frac{n(T+7)}{2T+13}\rceil-1$ and also length~$(n+3)T+19+n(T+7)-(T+6)\lceil\frac{n(T+7)}{2T+13}\rceil-1$. This gadget is connected on the starting trap to~$v_{3n+1}^M$. Notice that this gadget ends with a trap with reload time~$1$. This construction will be refered to as the \textit{blocking infrastructure}. We then connect the final trap with reload time~$1$ in the \textit{blocking infrastructure} to~$\tV$.
        \item We have the number~$\numAgents=(n+1)(100n^5T^5+100nT+(3T+20+3n))$ of assets and~$\goal=1$
        \item Finally, note that~$G$ can be constructed in polynomial time in~$T$ and~$n$, as there are only a polynomial number of verticies in these.
        
    \end{enumerate}
    
    \textbf{Conception}. The idea is to construct~$G$ such that it only allows batches of assets to pass through the \textit{limiting infrastructure} sequentially before a long delay. We then force each batch of assets to have~$T+7$ assets enter the \textit{final subtree}. This deactivates the traps on~$v_i^M$. The assets then can choose which \textit{element gadgets} to enter. This represents choosing elements for our subsets in \threePart. Finally we make sure that if and only if all the subtrees are filled is there a solution to \threePart. We ensure this via the \textit{blocking infrastructure}. Note when we refer to some assets exiting "during the final batch", it means that the final batch and the assets in the subtrees form a single long sequence, while "before the final batch" refers to when this doesn't happen (i.e. there is a gap somewhere). It should also be noted that the final batch is artificial, meaning it doesn't represent a subset, but rather is used to force the previous batch to enter the subtrees, otherwise they could start trying to mode through the \textit{destruction gadget} without entering the subtrees.

    Note that a sequence of~$x$ traps starting with reload time~$x$ and ending at reload time~$1$ can only be passed by 1 asset iff there are~$x+1$ assets sequentially moving forward. This will be used to enforce the mentioned long delays and to more easily think about the correctness of this construction.

  \textbf{Correctness}.
We start by stating a useful observation.
\begin{observation}\label{obs:boatsThroughSequentialSequence}
  If there is a sequence of~$c(v), c(v)-1, \dots, 1$ traps, an asset is able to pass this sequence if and only if
  there are~$c(v)+1$~assets moving one-by-one sequentially for some~$c(v)+1$
  time steps.
\end{observation}

    First, notice that~$3T+20+3n$ assets pass the
    \textit{limiting infrastructure} before a~$100n^5T^5+100nT$ time delay. We
    call the~$3T+20+3n$ assets that pass a \textit{batch} of assets. The last of
    these is called the \textit{final batch} of assets. Based on~$\numAgents$,
    the \textit{limiting infrastructure} passes~$n+1$ batches, each consisting
    of~$3T+20+3n$ assets. To prove correctness, we split the claim into lemmas.
    
    \begin{lemma}\label{lem:unfilled-subtree}
        If during the final batch not all non-trap nodes in the subtrees are filled, then~$\mathcal{J}$ is a \No-instance
    \end{lemma}
    \begin{claimproof}
        The \textit{blocking infrastructure} has length~$(n+3)T+19+n(T+7)-(T+6)\lceil\frac{n(T+7)}{2T+13}\rceil-1$. If even a single non-trap node is unfilled in the subtrees, then at most there can be~$(n+3)T+19+n(T+7)-(T+6)\lceil\frac{n(T+7)}{2T+13}\rceil-1$ assets that arrive at the \textit{blocking infrastructure}, by summing the number of remaining assets and calculating the loss of assets when exiting each\textit{subtree} and adding to this the number of assets in the final batch. By~\Cref{obs:boatsThroughSequentialSequence}, this means that there will be no successful assets.
        
        For the \textit{element gadgets} this calculation is obvious, each \textit{element gadget} loses one asset when exiting. 
        
        For the \textit{final subtree}, let~$e$ be the number of assets inside this subtree by the final batch. Say that there are~$e<n(T+7)$ assets and they are guarded by~$T+6$ traps. We need to show that for any~$e$, there is no case where an equal or greater number of assets can exit compared to when~$e=n(T+7)$. The assets can be processed in groups of~$2T+13$, since every~$T+6$ assets allow~$T+7$ assets to pass. Now it remains to show that~$n(T+7) \mod(2T+13) > T+6$, that is losing even a single asset means less than the desired number of assets can exit this subtree. Since~$n$ is odd, we write this as~$2p+1$, so~$(2p+1)(T+7)\mod(2T+13)$. Expanding gives us~$2kT+T+14k+7$, and the modulo is~$T+k+7$. Now, since~$n<T$ by our variant or \threePart,~$T+6<T+k+7<2T+13$ for~$k \geq 1$. Meaning that losing even a single asset results in a loss in the possible number of assets that can exit. 
        
        This means that having a single subtree that is not filled will result in a \No-instance.
    \end{claimproof}

    \begin{lemma}\label{lem:enter-final-subtree}
        If~$\mathcal{J}$ is a \Yes-instance of \basicProb, then at each round~$T+7$ assets enter the \textit{final subtree}
    \end{lemma}
    \begin{claimproof}
        Consider the contrapositive, that if at some round a number of assets not equal to~$T+7$ enter the \textit{final subtree}, then~$\mathcal{J}$ is a \No-instance of \basicProb.
        
        For each batch, at most~$T+7$ assets can enter the empty nodes of the \textit{final subtree}, since there are~$3T+20+3n$ assets per batch, the number of assets that reach the \textit{final subtree} is~$3T+20+3n-(3n+1)=3T+19$. Since there is a path of~$T+6$ traps on this subtree, these must be deactivated before allowing~$T+7$ assets into the empty nodes. But then the path of traps will reactivated one after another, needing another~$T+6$ assets before allowing any assets to enter the empty nodes.

        Now if at some round~$<T+7$ assets enter this subtree, then the subtree will not be full, resulting in a \No-instance by~\Cref{lem:unfilled-subtree}.
    \end{claimproof}
    \begin{lemma}\label{lem:agents-dont-leave-subtree}
        Any \Yes-instance in which assets have entered subtrees cannot have assets leaving subtrees before the final batch. Nor can it have a \textit{guard vertex} that is activated twice before the final batch.
    \end{lemma}
    \begin{claimproof}
        Assume there is some asset who either leaves the subtrees before the final batch or some \textit{guard vertex} is deactivated more than twice. Incorporating this into the lower bound, we find that we lose more assets than allowed from the~$nT$ term in the required~$(n+3)T+18+n(T+7)-(T+6)\lceil\frac{n(T+7)}{2T+13}\rceil$ assets to pass.
    \end{claimproof}
    \begin{lemma}\label{lem:perbatch-loss-limit}
        Each batch that passes through the \textit{limiting infrastructure} can lose at most~$T+10+3n$ before exiting during the final batch
    \end{lemma}
    \begin{claimproof}
        There are~$n$ batches before the final batch and~$(3T+20+3n)$ assets per batch. We lose~$3n+1$ assets activating all~$v_i^M$ and~$T+6$ assets activating the \textit{final subtree}. Another~$T+7$ assets enter the \textit{final subtree}. So we only have~$T+6$ assets remaining which must enter the \textit{element gadgets}. We have~$3n$ \textit{element gadgets}. By the restriction~$T/4<s(y)<T/2$, we must choose exactly 3 to enter, otherwise for all the subtrees to be full a trap must be activated more than twice, which doesn't happen according to~\Cref{lem:agents-dont-leave-subtree}.
    \end{claimproof}

    \begin{lemma}\label{lem:enter-element-gadget}
         If~$\mathcal{J}$ is a \Yes-instance then all assets enter
         \textit{element gadgets} before the next batch arrives
    \end{lemma}
    \begin{claimproof}
        Since~$v_i^M$ is much shorter than the delay between batches, the assets must enter subtrees to survive till the next batch, otherwise we lose more assets than allowed for some batch.
    \end{claimproof}
    
    Notice the general idea is to set it up so that even an additional loss of one asset in aggregate, or if one of the subtrees is not filled will result in~$0$ assets being able to pass.

    The above lemmas show that all the assets are only pushed through when the final batch of assets arrives.

   $(\rightarrow)$ If~$\mathcal{J}$ is a \Yes-instance, then~$\mathcal{P}$ is
   also a \Yes-instance. If we have a \Yes-instance, then all the non-trap nodes
   in the subtrees are filled by contraposition of~\Cref{lem:unfilled-subtree}.
   For each batch before the final batch,~$T+7$ assets will enter the
   \textit{final subtree} by~\Cref{lem:enter-final-subtree}. This will
   deactivate all~$v_i^M$. From this, we always choose 3 subtrees as discussed
   in~\Cref{lem:perbatch-loss-limit} and~\Cref{lem:enter-element-gadget}. So,
   for each batch, the 3 subtrees that are chosen will correspond to a~$Y_i$
   in~$\mathcal{P}$. Since these assets will only leave during the \textit{final
   batch} by~\Cref{lem:agents-dont-leave-subtree}, no element~$y\in Y$ can be
   chosen more than once and the union of~$Y_i$ is~$Y$. And since this applies
   for each batch, we have a \Yes-instance to~$\mathcal{P}$.
    
   $(\leftarrow)$ If~$\mathcal{P}$ is a \Yes-instance to \threePart,
   then~$\mathcal{J}$ is also a \Yes-instance. For each batch, we can enter the
   3 \textit{element gadgets} corresponding to the 3 chosen elements in
   \threePart. Since we have a \Yes-instance to \threePart, by the final batch
   all the subtrees will be filled by this plan. Then, in total we
   have~$(n+3)T+18+n(T+7)-(T+6)\lceil\frac{n(T+7)}{2T+13}\rceil$ assets that
   reach the \textit{blocking infrastructure}, which gives us one successful
   asset.
 \end{proof}

\restatetheorem{\thmnoapx*}
\begin{proof}
  \backtotheorem{thm:noapx}%
  Consider any instance~$\mathcal{P}$ of \threePart and construct an instance
  of \optBasicProb~$\mathcal{J}$ to be the corresponding instance
  of~$\mathcal{P}$ described in~\Cref{thm:np-hardness-trees}. Now assume that
  there is a constant factor approximation algorithm that runs in polynomial
  time. Since~$\goal=1$ and the assets are discrete, any constant factor
  approximation algorithm would have to result in at least~$1$ successful
  asset, including on~$\mathcal{J}$. But this would solve \threePart in
  polynomial time, which is a contradiction unless~$P=NP$.
\end{proof}

\section{\APXh{}ness of \optBasicProb}\label{app:apx-hardness}

Applying a slightly different reduction, we show \optBasicProb{} is generally
\APXh{}. However, our proof (inspired by an \APXhness{} proof
by~\citet{Viglietta2015}) drops the structural restrictions that the topology is
a tree of maximum degree~$3$.
\begin{theorem}\label{thm:apx-hardness}
    \optBasicProb is \APXh.
\end{theorem}
\begin{proof}
  For the sake of clarity and completeness, let us start with a formal
  definition of~$\optBasicProb$.

  \defProblemQuestion{\optBasicProb}%
  {An undirected graph~$G = (V,E)$, an initial and target vertex~$\sV,\tV \in
   V$, a set of traps~$\traps$, a reload time function~$\reloadFn\colon \traps
   \to \naturals$, and a set of assets~$\agents$.}%
  {What is the maximum number of successful assets?}

  In our proof, we will reduce from \maxThreeSatThree{}, which is \APXh{} (for a
  detailed and accessible treatment of \maxThreeSatThree{}, see
  a draft of the book by~\citet[see
  p.~226,228]{erikd.demaineComputationalIntractabilityGuide2025}). Here, given a
  constrained SAT formula $\varphi$ where every clause has at most $3$ literals
  and every variable occurs at most $3$~times, we ask for the maximum number of
  clauses that can be satisfied by an assignment.

  Let~$\mathcal{P}$ be an instance of \maxThreeSatThree and~$\mathcal{J} =
  (G,\sV,\tV,\traps,\reloadFn,\agents,\goal)$ be an instance of \optBasicProb.
  We show an \LRed from~$\mathcal{P}$ to~$\mathcal{J}$. Let the
  number of variables in~$\mathcal{P}$ be~$V$ and the number of clauses
  be~$C$. Note that we set~$V\geq 10$ so we can easily think about the construction in
  terms of degree.
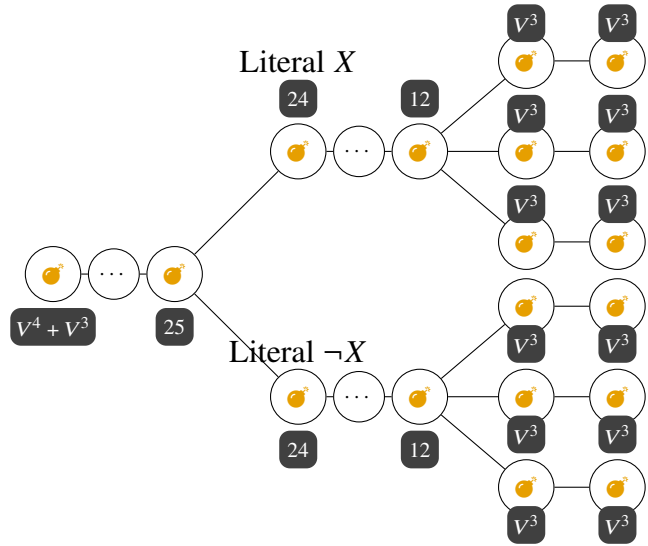
\begin{figure}
  \centering
    \begin{tikzpicture}[node distance=0.8cm,minimum size=0.5cm,font=\footnotesize]%
    \node[right of=s,circle, draw] (ds1x) {{\textcolor{cbOrange}{\faIcon{bomb}}}};
    \node[reload,below of=ds1x,node distance=0.7cm] {$V^4+V^3$};
    
    \node[right of=ds1x,circle, draw] (ds2) {\dots};

    \node[right of=ds2,circle, draw] (dsEnd) {{\textcolor{cbOrange}{\faIcon{bomb}}}};
    \node[reload,below of=dsEnd,node distance=0.7cm] {$25$};

    \node[above right of=dsEnd,circle,node distance=2.3cm,draw] (a1) {{\textcolor{cbOrange}{\faIcon{bomb}}}};
    \node[reload,above of=a1,node distance=0.7cm] (loc) {$24$};

    \node[above of=loc,node distance=0.5cm, font=\large] (literal) {Literal~$X$};

    \node[right of=a1,circle,draw] (a2) {\dots};

    \node[right of=a2,circle,draw] (a3) {{\textcolor{cbOrange}{\faIcon{bomb}}}};
    \node[reload,above of=a3,node distance=0.7cm] {$12$};

    \node[right of=a3,circle,node distance=1.4cm,draw] (va1) {{\textcolor{cbOrange}{\faIcon{bomb}}}};
    \node[above of=va1,node distance=0.5cm,reload] {$V^3$};

    \node[above of=va1,circle,node distance=1.2cm,draw] (va2) {{\textcolor{cbOrange}{\faIcon{bomb}}}};
    \node[above of=va2,node distance=0.5cm,reload] {$V^3$};

    \node[below of=va1,circle,node distance=1.2cm,draw] (va3) {{\textcolor{cbOrange}{\faIcon{bomb}}}};
    \node[above of=va3,node distance=0.5cm,reload] {$V^3$};
    \node[right of=va1,circle,draw, node distance=1.2cm] (ax1) {{\textcolor{cbOrange}{\faIcon{bomb}}}};
    \node[above of=ax1,node distance=0.5cm,reload] {$V^3$};

    \node[right of=va2,circle,draw, node distance=1.2cm] (ax2) {{\textcolor{cbOrange}{\faIcon{bomb}}}};
    \node[above of=ax2,node distance=0.5cm,reload] {$V^3$};

    \node[right of=va3,circle,draw, node distance=1.2cm] (ax3) {{\textcolor{cbOrange}{\faIcon{bomb}}}};
    \node[above of=ax3,node distance=0.5cm,reload] {$V^3$};

    \node[below right of=dsEnd,circle,node distance=2.3cm,draw] (b1) {{\textcolor{cbOrange}{\faIcon{bomb}}}};
    \node[reload,below of=b1,node distance=0.7cm] (loc) {$24$};
    \node[above of=loc,node distance=1.3cm, font=\large] (literal) {Literal~$\neg X$};

    \node[right of=b1,circle,draw] (b2) {\dots};

    \node[right of=b2,circle,draw] (b3) {{\textcolor{cbOrange}{\faIcon{bomb}}}};
    \node[reload,below of=b3,node distance=0.7cm] {$12$};
    
    \node[right of=b3,circle,node distance=1.4cm,draw] (vb1) {{\textcolor{cbOrange}{\faIcon{bomb}}}};
    \node[below of=vb1,node distance=0.5cm,reload] {$V^3$};

    \node[above of=vb1,circle,node distance=1.2cm,draw] (vb2) {{\textcolor{cbOrange}{\faIcon{bomb}}}};
    \node[below of=vb2,node distance=0.5cm,reload] {$V^3$};

    \node[below of=vb1,circle,node distance=1.2cm,draw] (vb3) {{\textcolor{cbOrange}{\faIcon{bomb}}}};
    \node[below of=vb3,node distance=0.5cm,reload] {$V^3$};

    \node[right of=vb1,circle,draw, node distance=1.2cm] (bx1) {{\textcolor{cbOrange}{\faIcon{bomb}}}};
    \node[below of=bx1,node distance=0.5cm,reload] {$V^3$};

    \node[right of=vb2,circle,draw, node distance=1.2cm] (bx2) {{\textcolor{cbOrange}{\faIcon{bomb}}}};
    \node[below of=bx2,node distance=0.5cm,reload] {$V^3$};

    \node[right of=vb3,circle,draw, node distance=1.2cm] (bx3) {{\textcolor{cbOrange}{\faIcon{bomb}}}};
    \node[below of=bx3,node distance=0.5cm,reload] {$V^3$};

    \draw (ds1x) -- (ds2);
    \draw (ds2) -- (dsEnd);
    \draw (dsEnd) -- (a1);
    \draw (dsEnd) -- (b1);
    \draw (a1) -- (a2);
    \draw (a2) -- (a3);
    \draw (b1) -- (b2);
    \draw (b2) -- (b3);
    \draw (a3) -- (va1);
    \draw (a3) -- (va2);
    \draw (a3) -- (va3);
    \draw (b3) -- (vb1);
    \draw (b3) -- (vb2);
    \draw (b3) -- (vb3);
    \draw (va1) -- (ax1);
    \draw (va2) -- (ax2);
    \draw (va3) -- (ax3);
    \draw (vb1) -- (bx1);
    \draw (vb2) -- (bx2);
    \draw (vb3) -- (bx3);
    \end{tikzpicture}
  \caption{Variable Gadget for APX-hardness reduction. }
  \label{fig:variable-gadget}
\end{figure}
    
\begin{figure}
  \centering
    \begin{tikzpicture}[node distance=1.5cm,minimum size=0.5cm, font=\footnotesize]
    \node[circle, draw] (t1) {{\textcolor{cbOrange}{\faIcon{bomb}}}};
    \node[below of=t1,node distance=0.7cm,reload] {$4C$};
    
    \node[right of=t1,circle, draw] (t2) {{\textcolor{cbOrange}{\faIcon{bomb}}}};
    \node[below of=t2,node distance=0.7cm,reload] {$4C$};

    \node[right of=t2,circle, draw] (t3) {{\textcolor{cbOrange}{\faIcon{bomb}}}};
    \node[below of=t3,node distance=0.7cm,reload] {$1$};

    \node[right of=t3,circle, draw] (t4) {{\textcolor{cbOrange}{\faIcon{bomb}}}};
    \node[below of=t4,node distance=0.7cm,reload] {$1$};

    \node[above of=t4,circle, draw] (tv1) {{\textcolor{cbOrange}{\faIcon{bomb}}}};
    \node[right of=tv1,node distance=0.7cm,reload] {$1$};
    
    \node[above of=tv1,rectangle, draw] (tv2) {Literal 2};

    \node[right of=tv2,rectangle, draw] (tv3) {Literal 3};

    \node[left of=tv2,rectangle, draw] (tv4) {Literal 1};

    \node[right of=t4,circle, draw] (t5) {};

    \draw (t1) -- (t2);
    \draw (t2) -- (t3);
    \draw (t3) -- (t4);
    \draw (t4) -- (tv1);
    \draw (tv1) -- (tv2);
    \draw (tv1) -- (tv3);
    \draw (tv1) -- (tv4);
    \draw (t4) -- (t5);
    
    \end{tikzpicture}
  \caption{Clause Gadget for APX-hardness reduction.}
  \label{fig:clause-gadget}
\end{figure}
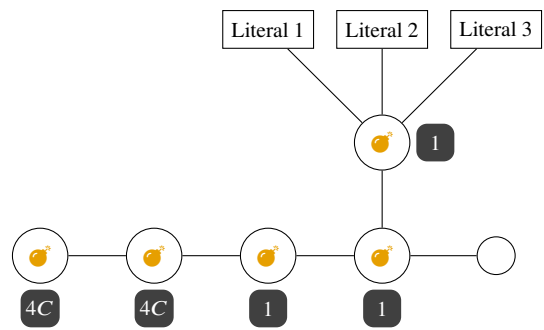
\begin{figure*}
  \centering
    \begin{tikzpicture}[node distance=0.7cm,minimum size=0.3cm, font=\tiny]
    \node[very thick,circle, draw] (s) {$s$};

    \node[right of=s,circle, draw, node distance=10cm] (fin1) {{\textcolor{cbOrange}{\faIcon{bomb}}}};
    \node[above of=fin1,node distance=0.5cm,reload] {$3V+C-1$};

    \node[above right of=s,circle, draw,node distance=1.5cm] (ds1x) {{\textcolor{cbOrange}{\faIcon{bomb}}}};
    \node[above of=ds1x,node distance=0.5cm,reload] {$V^4+V^3$};
    \draw (s) -- (ds1x);
    \node[right of=ds1x,circle, draw] (ds2) {\dots};

    \node[right of=ds2,circle, draw] (dsEnd) {{\textcolor{cbOrange}{\faIcon{bomb}}}};
    \node[above of=dsEnd,node distance=0.5cm,reload] {$25$};

    \node[above right= 1cm and 0.5cm of dsEnd,circle,draw] (a1) {{\textcolor{cbOrange}{\faIcon{bomb}}}};
    \node[above of=a1,node distance=0.5cm,reload] (loc) {$24$};
    \node[above of=loc,node distance=0.5cm, font=\normalsize] (literal) {$X_3$};

    \node[right of=a1,circle,draw] (a2) {\dots};

    \node[right of=a2,circle,draw] (a3) {{\textcolor{cbOrange}{\faIcon{bomb}}}};
    \node[above of=a3,node distance=0.5cm,reload] {$12$};

    \node[right of=a3,circle,node distance=1cm,draw] (va1) {{\textcolor{cbOrange}{\faIcon{bomb}}}};
    \node[above of=va1,node distance=0.4cm,reload] {$V^3$};

    \node[above of=va1,circle,node distance=1cm,draw] (va2) {{\textcolor{cbOrange}{\faIcon{bomb}}}};
    \node[above of=va2,node distance=0.4cm,reload] {$V^3$};

    \node[below of=va1,circle,node distance=1cm,draw] (va3) {{\textcolor{cbOrange}{\faIcon{bomb}}}};
    \node[above of=va3,node distance=0.4cm,reload] {$V^3$};

    \node[right of=va1,circle,draw, node distance=1cm] (ax11) {{\textcolor{cbOrange}{\faIcon{bomb}}}};
    \node[above of=ax11,node distance=0.4cm,reload] {$V^3$};

    \node[right of=va2,circle,draw, node distance=1cm] (ax21) {{\textcolor{cbOrange}{\faIcon{bomb}}}};
    \node[above of=ax21,node distance=0.4cm,reload] {$V^3$};

    \node[right of=va3,circle,draw, node distance=1cm] (ax31) {{\textcolor{cbOrange}{\faIcon{bomb}}}};
    \node[above of=ax31,node distance=0.4cm,reload] {$V^3$};

    \node[below right= 1cm and .5cm of dsEnd,circle,draw] (b1) {{\textcolor{cbOrange}{\faIcon{bomb}}}};
    \node[below of=b1,node distance=0.5cm,reload] (loc) {$24$};
    \node[above of=loc,node distance=1.2cm, font=\normalsize] (literal) {$\neg X_3$};

    \node[right of=b1,circle,draw] (b2) {\dots};

    \node[right of=b2,circle,draw] (b3) {{\textcolor{cbOrange}{\faIcon{bomb}}}};
    \node[below of=b3,node distance=0.5cm,reload] {$12$};
    
    \node[right of=b3,circle,node distance=1cm,draw] (vb1) {{\textcolor{cbOrange}{\faIcon{bomb}}}};
    \node[below of=vb1,node distance=0.4cm,reload] {$V^3$};

    \node[above of=vb1,circle,node distance=1cm,draw] (vb2) {{\textcolor{cbOrange}{\faIcon{bomb}}}};
    \node[below of=vb2,node distance=0.4cm,reload] {$V^3$};

    \node[below of=vb1,circle,node distance=1cm,draw] (vb3) {{\textcolor{cbOrange}{\faIcon{bomb}}}};
    \node[below of=vb3,node distance=0.4cm,reload] {$V^3$};

    \node[right of=vb1,circle,draw, node distance=1cm] (bx11) {{\textcolor{cbOrange}{\faIcon{bomb}}}};
    \node[below of=bx11,node distance=0.4cm,reload] {$V^3$};

    \node[right of=vb2,circle,draw, node distance=1cm] (bx21) {{\textcolor{cbOrange}{\faIcon{bomb}}}};
    \node[below of=bx21,node distance=0.4cm,reload] {$V^3$};

    \node[right of=vb3,circle,draw, node distance=1cm] (bx31) {{\textcolor{cbOrange}{\faIcon{bomb}}}};
    \node[below of=bx31,node distance=0.4cm,reload] {$V^3$};

    \draw (ds1x) -- (ds2);
    \draw (ds2) -- (dsEnd);
    \draw (dsEnd) -- (a1);
    \draw (dsEnd) -- (b1);
    \draw (a1) -- (a2);
    \draw (a2) -- (a3);
    \draw (b1) -- (b2);
    \draw (b2) -- (b3);
    \draw (a3) -- (va1);
    \draw (a3) -- (va2);
    \draw (a3) -- (va3);
    \draw (b3) -- (vb1);
    \draw (b3) -- (vb2);
    \draw (b3) -- (vb3);
    \draw (va1) -- (ax11);
    \draw (va2) -- (ax21);
    \draw (va3) -- (ax31);
    \draw (vb1) -- (bx11);
    \draw (vb2) -- (bx21);
    \draw (vb3) -- (bx31);

    \node[above of=ds1x,rectangle, draw,node distance=3.5cm] (ds1) {Var 2};
    \draw (s) -- (ds1);

    \node[above of=ds1,circle, draw,node distance=3.5cm] (ds1) {{\textcolor{cbOrange}{\faIcon{bomb}}}};
    \node[above of=ds1,node distance=0.5cm,reload] {$V^4+V^3$};
    \draw (s) -- (ds1);
    \node[right of=ds1,circle, draw] (ds2) {\dots};

    \node[right of=ds2,circle, draw] (dsEnd) {{\textcolor{cbOrange}{\faIcon{bomb}}}};
    \node[above of=dsEnd,node distance=0.5cm,reload] {$25$};

    \node[above right = 1cm and .5cm of dsEnd,circle,draw] (a1) {{\textcolor{cbOrange}{\faIcon{bomb}}}};
    \node[above of=a1,node distance=0.5cm,reload] (loc) {$24$};
    \node[above of=loc,node distance=0.5cm, font=\normalsize] (literal) {$X_1$};

    \node[right of=a1,circle,draw] (a2) {\dots};

    \node[right of=a2,circle,draw] (a3) {{\textcolor{cbOrange}{\faIcon{bomb}}}};
    \node[above of=a3,node distance=0.5cm,reload] {$12$};

    \node[right of=a3,circle,node distance=1cm,draw] (va1) {{\textcolor{cbOrange}{\faIcon{bomb}}}};
    \node[above of=va1,node distance=0.4cm,reload] {$V^3$};

    \node[above of=va1,circle,node distance=1cm,draw] (va2) {{\textcolor{cbOrange}{\faIcon{bomb}}}};
    \node[above of=va2,node distance=0.4cm,reload] {$V^3$};

    \node[below of=va1,circle,node distance=1cm,draw] (va3) {{\textcolor{cbOrange}{\faIcon{bomb}}}};
    \node[above of=va3,node distance=0.4cm,reload] {$V^3$};

    \node[right of=va1,circle,draw, node distance=1cm] (ax13) {{\textcolor{cbOrange}{\faIcon{bomb}}}};
    \node[above of=ax13,node distance=0.4cm,reload] {$V^3$};

    \node[right of=va2,circle,draw, node distance=1cm] (ax23) {{\textcolor{cbOrange}{\faIcon{bomb}}}};
    \node[above of=ax23,node distance=0.4cm,reload] {$V^3$};

    \node[right of=va3,circle,draw, node distance=1cm] (ax33) {{\textcolor{cbOrange}{\faIcon{bomb}}}};
    \node[above of=ax33,node distance=0.4cm,reload] {$V^3$};

    \node[below right= 1cm and 0.5cm of dsEnd,circle,draw] (b1) {{\textcolor{cbOrange}{\faIcon{bomb}}}};
    \node[below of=b1,node distance=0.5cm,reload] (loc) {$24$};
    \node[above of=loc,node distance=1.2cm, font=\normalsize] {$\neg X_1$};

    \node[right of=b1,circle,draw] (b2) {\dots};

    \node[right of=b2,circle,draw] (b3) {{\textcolor{cbOrange}{\faIcon{bomb}}}};
    \node[below of=b3,node distance=0.5cm,reload] {$12$};
    
    \node[right of=b3,circle,node distance=1cm,draw] (vb1) {{\textcolor{cbOrange}{\faIcon{bomb}}}};
    \node[below of=vb1,node distance=0.4cm,reload] {$V^3$};

    \node[above of=vb1,circle,node distance=1cm,draw] (vb2) {{\textcolor{cbOrange}{\faIcon{bomb}}}};
    \node[below of=vb2,node distance=0.4cm,reload] {$V^3$};

    \node[below of=vb1,circle,node distance=1cm,draw] (vb3) {{\textcolor{cbOrange}{\faIcon{bomb}}}};
    \node[below of=vb3,node distance=0.4cm,reload] {$V^3$};

    \node[right of=vb1,circle,draw, node distance=1cm] (bx13) {{\textcolor{cbOrange}{\faIcon{bomb}}}};
    \node[below of=bx13,node distance=0.4cm,reload] {$V^3$};

    \node[right of=vb2,circle,draw, node distance=1cm] (bx23) {{\textcolor{cbOrange}{\faIcon{bomb}}}};
    \node[below of=bx23,node distance=0.4cm,reload] {$V^3$};

    \node[right of=vb3,circle,draw, node distance=1cm] (bx33) {{\textcolor{cbOrange}{\faIcon{bomb}}}};
    \node[below of=bx33,node distance=0.4cm,reload] {$V^3$};

    \draw (ds1) -- (ds2);
    \draw (ds2) -- (dsEnd);
    \draw (dsEnd) -- (a1);
    \draw (dsEnd) -- (b1);
    \draw (a1) -- (a2);
    \draw (a2) -- (a3);
    \draw (b1) -- (b2);
    \draw (b2) -- (b3);
    \draw (a3) -- (va1);
    \draw (a3) -- (va2);
    \draw (a3) -- (va3);
    \draw (b3) -- (vb1);
    \draw (b3) -- (vb2);
    \draw (b3) -- (vb3);
    \draw (va1) -- (ax13);
    \draw (va2) -- (ax23);
    \draw (va3) -- (ax33);
    \draw (vb1) -- (bx13);
    \draw (vb2) -- (bx23);
    \draw (vb3) -- (bx33);
    
    \draw[decorate,decoration={brace,amplitude=10pt,raise=1cm}]
    (ds1x.west) -- (ds1.west)
    node[midway] (brace) {};
    \node[below of=brace,xshift=-1.7cm,node distance=0cm, rotate=90,
      font=\normalsize] {Variable Gadgets};

    \node[draw, circle, below right = 4cm and 2cm of s] (x1){{\textcolor{cbOrange}{\faIcon{bomb}}}};
    \node[below of=x1,node distance=0.5cm,reload] {$V^6+V^5+3V+4C$};

    \node[draw, circle, right of=x1] (x2){\dots};
    \node[draw, circle, right of=x2] (x3){{\textcolor{cbOrange}{\faIcon{bomb}}}};
    \node[above of=x3,node distance=0.5cm,reload] {$V^3+V^2+3V+4C$};  
    \draw (s) -- (x1);
    \draw (x1) -- (x2);
    \draw (x2) -- (x3);

    \node[draw, circle, below of=x3, node distance=4cm] (z1){{\textcolor{cbOrange}{\faIcon{bomb}}}};
    \node[below of=z1,node distance=0.5cm,reload] {$V^3+V^2+3V$};

    \node[draw, circle, right of=z1, node distance=2.5cm] (z2){\dots};

    \node[draw, circle, right of=z2, node distance=1.6cm] (z3){{\textcolor{cbOrange}{\faIcon{bomb}}}};
    \node[below of=z3,node distance=0.5cm,reload] {$3V$};
    
    \draw (x3) -- (z1);
    \draw (z1) -- (z2);
    \draw (z2) -- (z3);
    \draw (z3) -- (fin1);

    \draw[decorate,decoration={mirror,brace,amplitude=10pt,raise=0.5cm}]
    (z1.south) -- (z3.south)
    node[midway] (brace) {};
    \node[below of=brace,node distance=0cm,yshift=-1cm, font=\normalsize] {Slowdown Path};

    \def\zzag{.2}
    \node[circle, draw, right of=x3, node distance=1.5cm] (tt1) {{\textcolor{cbOrange}{\faIcon{bomb}}}};
    \node[below of=tt1,node distance=0.5cm,reload] {$4C$};
    
    \node[below right= \zzag and \zzag of tt1,circle, draw] (t2) {{\textcolor{cbOrange}{\faIcon{bomb}}}};
    \node[below of=t2,node distance=0.5cm,reload] {$4C$};

    \node[above right= \zzag and \zzag of t2,circle, draw] (t3) {{\textcolor{cbOrange}{\faIcon{bomb}}}};
    \node[below of=t3,node distance=0.5cm,reload] {$1$};

    \node[below right= \zzag and \zzag of t3,circle, draw] (t4) {{\textcolor{cbOrange}{\faIcon{bomb}}}};
    \node[below of=t4,node distance=0.5cm,reload] {$1$};

    \node[above right= 10cm and 1.25cm of t4,circle, draw] (tv11) {{\textcolor{cbOrange}{\faIcon{bomb}}}};
    \node[right of=tv11,node distance=0.5cm,reload] {$1$};
    
    \node[above right= \zzag and {\zzag + .5} cm of t4,circle, draw] (t51) {};
    \draw (x3) -- (tt1);
    \draw (tt1) -- (t2);
    \draw (t2) -- (t3);
    \draw (t3) -- (t4);
    \draw (t4) -- (tv11);
    \draw (t4) -- (t51);
    \draw (fin1.west) -- +(-.6,0) -- (t51);
    \node[circle, draw, above of=tt1, node distance=1.75cm] (t1) {{\textcolor{cbOrange}{\faIcon{bomb}}}};
    \node[below of=t1,node distance=0.5cm,reload] {$4C$};
    
    \node[below right= \zzag and \zzag of t1,circle, draw] (t2) {{\textcolor{cbOrange}{\faIcon{bomb}}}};
    \node[below of=t2,node distance=0.5cm,reload] {$4C$};

    \node[above right= \zzag and \zzag of t2,circle, draw] (t3) {{\textcolor{cbOrange}{\faIcon{bomb}}}};
    \node[below of=t3,node distance=0.5cm,reload] {$1$};

    \node[below right= \zzag and \zzag of t3,circle, draw] (t4) {{\textcolor{cbOrange}{\faIcon{bomb}}}};
    \node[below of=t4,node distance=0.5cm,reload] {$1$};

    \node[above right= 10cm and 1cm of t4,circle, draw] (tv12) {{\textcolor{cbOrange}{\faIcon{bomb}}}};
    \node[right of=tv12,node distance=0.5cm,reload] {$1$};

    \node[above right= \zzag and {\zzag + .5} of t4,circle, draw] (t52) {};
    \draw (x3) -- (t1);
    \draw (t1) -- (t2);
    \draw (t2) -- (t3);
    \draw (t3) -- (t4);
    \draw (t4) -- (tv12);
    \draw (t4) -- (t52);
    \draw (fin1.north west) -- +(-.85,0) -- (t52);
    \node[circle, draw, below of=tt1, node distance=1.75cm] (t1) {{\textcolor{cbOrange}{\faIcon{bomb}}}};
    \node[below of=t1,node distance=0.5cm,reload] {$4C$};
    
    \node[below right= \zzag and \zzag of t1,circle, draw] (t2) {{\textcolor{cbOrange}{\faIcon{bomb}}}};
    \node[below of=t2,node distance=0.5cm,reload] {$4C$};

    \node[above right= \zzag and \zzag of t2,circle, draw] (t3) {{\textcolor{cbOrange}{\faIcon{bomb}}}};
    \node[below of=t3,node distance=0.5cm,reload] {$1$};

    \node[below right= \zzag and \zzag of t3,circle, draw] (t4) {{\textcolor{cbOrange}{\faIcon{bomb}}}};
    \node[below of=t4,node distance=0.5cm,reload] {$1$};

    \node[above right= 10cm and 1.5cm of t4,circle, draw] (tv13) {{\textcolor{cbOrange}{\faIcon{bomb}}}};
    \node[right of=tv13,node distance=0.5cm,reload] {$1$};

    \node[above right= \zzag and  {\zzag + .5} of t4,circle, draw] (t53) {};
    \draw (x3) -- (t1);
    \draw (t1) -- (t2);
    \draw (t2) -- (t3);
    \draw (t3) -- (t4);
    \draw (t4) -- (tv13);
    \draw (t4) -- (t53);
    \draw (fin1.south west) -- +(-0.5, 0) -- (t53);

    \draw[decorate,decoration={mirror,brace,amplitude=10pt,raise=1.25cm}]
    (t53.east) -- (t52.east)
    node[midway] (brace) {};
    \node[below of=brace,xshift=1.9cm,node distance=0cm, font=\normalsize,
      rotate=90] {Clause Gadgets};

    \node[draw, circle, right of=fin1] (fin2){\dots};
    \node[draw, circle, right of=fin2] (fin3){{\textcolor{cbOrange}{\faIcon{bomb}}}};
    \node[above of=fin3,node distance=0.5cm,reload] {$C$};  
    \draw (fin1) -- (fin2);
    \draw (fin2) -- (fin3);
    \node[draw, circle, very thick, right of=fin3](t){t};
    \draw (fin3) -- (t);
    \draw[blue] (tv11) -- (ax11);
    \draw[ blue] (tv11) -- (ax13);
    \draw[red] (tv13) -- (bx13);
    \draw[red] (tv13) -- (bx11);
    \draw[blue] (tv12) -- (ax23);
    \draw[red] (tv12) -- (bx21);

    \draw[decorate,decoration={mirror,brace,amplitude=10pt}]
    (fin1.south) -- (fin3.south)
    node[midway] (brace) {};
    \node[below of=brace,xshift=.05cm,yshift=-0.75cm,node distance=0cm,
      font=\normalsize] {\parbox{2cm}{\centering Blocking Infrastructure}};

    \end{tikzpicture}
  \caption{Construction for showing \APX-hardness. The formula for this construction is~$(X_1 \vee \neg X_3) \wedge (X_1 \vee X_3)\wedge (\neg X_1 \vee \neg X_3)$. The blue and red lines are used to differentiate positive and negative literals respectively.}
  \label{fig:APX-hardness-construction}
\end{figure*}
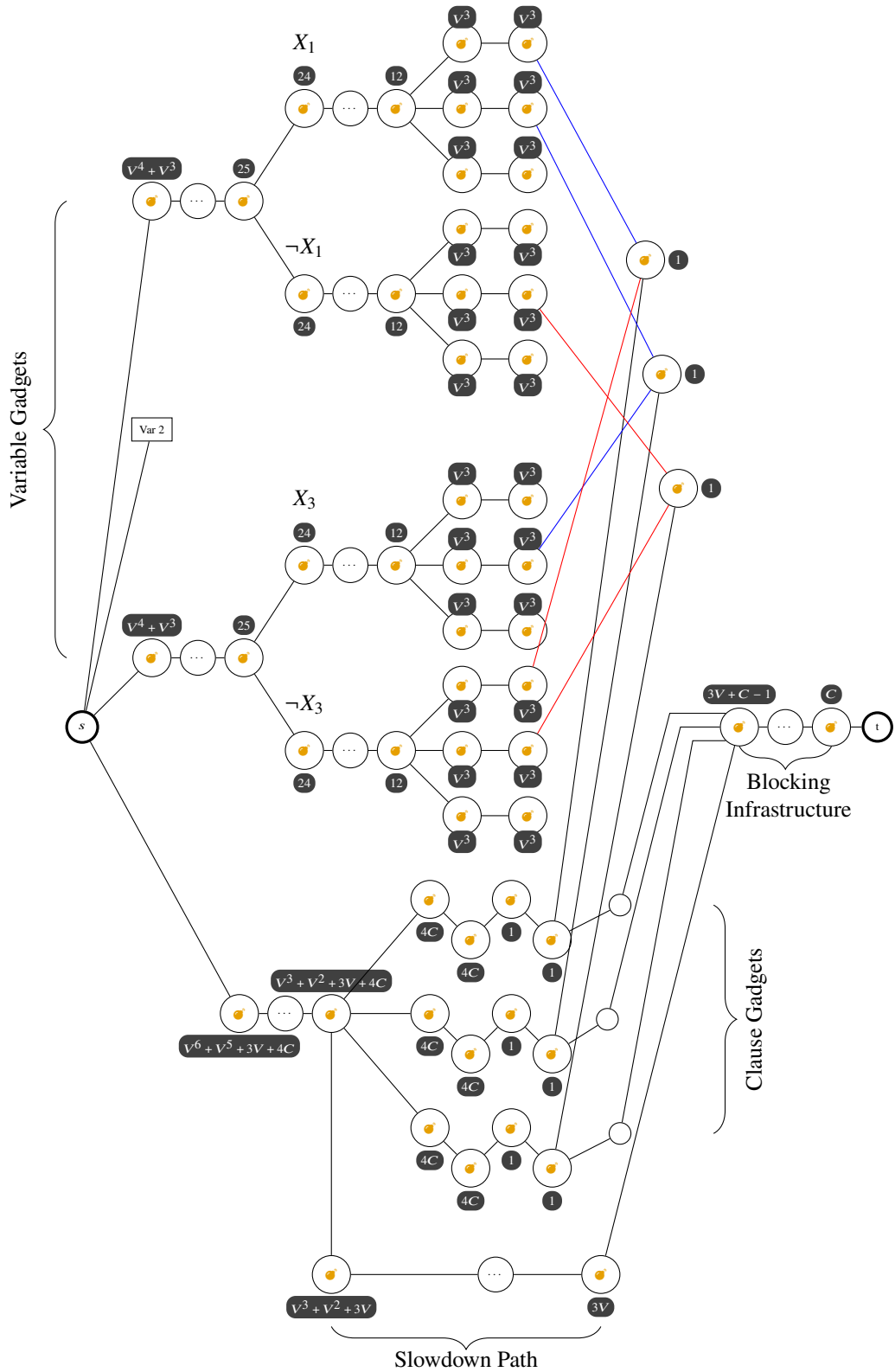
    \textbf{Construction}.
    We will use gadgets defined in~\Cref{thm:np-hardness-trees}. First, we define a few more gadgets and the infrastructure for the reduction
    \begin{enumerate}
        \item The \textit{variable gadget}~$v_{1}^i$ represents the~$i^{th}$ variable. It first consists of a \textit{destruction gadget} defined by starting trap with reload time~$V^4+V^3$ and length~$V^4+V^3-23$. The last trap of this gadget has reload time~$24$. Following this last trap, the path splits into 2 more paths, each a \textit{destruction gadget} starting with reload time 23 and length 12, ending with reload time 12. We call each of these end nodes. Following the end of both these paths are 3 length 2 paths, all with reload times~$V^3$, all connected to the end node. The \textit{variable gadget} is illustrated in~\Cref{fig:variable-gadget}
        \item The start of the \textit{clause gadget} is just 2 traps in a path with reload time~$4C$ in a path, followed by 3 reload time 1 traps also arranged in a path. The second of these is connected to a single empty node. The third of these will later be connected to the literals of the corresponding clause. This is illustrated in~\Cref{fig:clause-gadget}. 
        \item The \textit{blocking infrastructure} is a \textit{destruction gadget} that starts with a trap with a reload time of~$3V+C-1$ and has length~$3V$. This \textit{destruction gadget} ends with a trap with a reload time of~$C$.
        \item The \textit{slowdown path} is a \textit{destruction gadget} starting with reload time~$V^3+V^2+3V$ and ending with reload time~$3V$.
    \end{enumerate}
    Now we construct the full graph~$G$. A complete sample construction is shown in~\Cref{fig:APX-hardness-construction}. 
    \begin{enumerate}
        \item First, we connect a \textit{destruction gadget} with starting reload~$V^6+V^5+3V+4C$ and length~$V^6+V^5-V^3-V^2+1$. Note that it has ending reload time~$V^3+V^2+3V+4C$ to~$\sV$. We call this destruction gadget~$v^L$ and call the last node~$v_a^L$
        \item We also connect all the variable gadgets directly to~$\sV$ by the \textit{starting trap} of the first \textit{destruction gadget} of~$v_{1}^i$.
        \item We connect each of the 3 paths length 2 paths in each \textit{variable gadget} to the last reload time 1 trap of a \textit{clause gadget}. This connection is made for all the corresponding literals of the clause in our original problem that the \textit{clause gadget} represents.
        \item  We then connect to~$v_a^L$ the~$C$ clause gadgets (on the first~$4C$ reload time trap) and the \textit{slowdown path} (on the \textit{starting trap})
        \item The last node of the \textit{slowdown path} and the empty node of all the clause gadgets are connected to the starting trap of the \textit{blocking infrastructure}. 
        \item Finally we set~$\numAgents=V^6+V^5+3V+4C+1+V(V^4+V^3+1)=V^6+2V^5+V^4+4V+4C+1$; Each variable gadget requires~$V^4+V^3+1$ assets, and~$V^6+V^5+3V+4C+1$ assets go through~$v^L$
        \item This construction takes polynomial time in~$C$ and~$V$ because the number of nodes is polynomial in these variables.

    \end{enumerate}

    \textbf{Conception}. Choosing the path to take here corresponds to choosing an assignment for the variable. Note that you never want to split the vertices so that you enter both paths at the same time or reroute assets from some variable gadget to another. These variable gadgets then will allow the empty node in the clause gadgets to enter. But because of the \textit{slowdown path}, the assets have to wait a long time there, by which time all assets not on the empty node will have died and the activated traps reset. This limits each clause gadget to contributing to only a single successful assets. Finally the assets that pass through the \textit{slowdown path} will deactivate all the traps on the \textit{blocking infrastructure}. allowing any assets previously on the empty nodes to pass.
    
    \textbf{Correctness}.
    Now, we move on to the proof. We first prove the following lemmas.
    \begin{lemma}
        The \textit{slowdown path} is deactivated 
    \end{lemma}
    \begin{claimproof}
        Assume the \textit{slowdown path} is not used. We set the \textit{destruction gadget}~$v^L$ to have length~$V^6+V^5-V^3-V^2-1>V(V^4+V^3)+V^3+V^2+3V$, which holds for~$V\geq 3$. This means we cannot reroute the assets to pass~$v^L$ twice in 2 batches. We cannot pass any variable gadgets either as~$V^4+V^3>V^3+V^2+3V$. So we have an additions~$V^3+V^2+3V$ assets which can pass the \textit{clause gadgets}. By a counting argument, since to pass an asset through the clause 2 assets from the \textit{variable gadget} are needed, at most~$3V$ assets can pass from the \textit{clause gadget} to the \textit{blocking infrastructure}. We have~$0$ assets coming from the \textit{slowdown path} since by assumption it is not used. Since we need at least~$3V+1$ assets to arrive at the \textit{blocking infrastructure}, for any asset to reach~$\tV$, we have~$\goal=0$
    \end{claimproof}

    \begin{lemma}
        Some \textit{variable gadget} and some \textit{clause gadget} is used (deactivated)
    \end{lemma}
    \begin{claimproof}
        Assume the \textit{variable gadget} is not used. Then~$0$ assets can pass through the \textit{clause gadget} because of the 2 reload time~$1$ traps. There are also no assets that make it through the \textit{slowdown path}, as at most~$3V$ assets arrive here at a time before a long delay, but there are also~$3V$ traps on the \textit{blocking infrastructure}.

        The proof for the case of no \textit{clause gadget} being deactivated is the same.
    \end{claimproof}
 
    Next, the \textit{variable gadget} can only choose one of either a true assignment or a false assignment. Since if we assume this is not the case, after the choosing of the first assignment, the~$V^3$ trap reload time must be used before the next group of assets arrives for the second assignment. And this means the traps on the \textit{clause gadgets} must also have assets before the reload time of the variable gadget is reset. By the time the next group of assets arrives for the assignment of the variables, the traps of the \textit{clause gadget} will already be reset and there will be no assets that can pass again.
    
    Now, we show that each clause can only contribute one asset, which is stored on the final regular vertex of the \textit{clause gadget}. Assume that we can contribute more than one asset from a clause gadget. Then the assets must pass through the 3 reload time 1 traps. But afterwards, because of the very long \textit{slowdown path} we must wait first until the assets on that path are sufficiently close. But then the trap will reactivate and kill one of the assets, since there is no place to store them. Neither can we wait at the start of the clause gadget, for the same reason; by the time the slowdown path reaches the end, the trap will have been reset.
    
    This reduction preserves the optimal value since any plan that has a maximum of~$k$ successful trials can be converted to an assignment that has exactly~$k$ true clauses by setting the corresponding literals and vice versa. Because of this one-to-one mapping between \optBasicProb instances with a max of~$k$ successful trials and \maxThreeSatThree with a max of~$k$ true clauses (i.e., you can convert~$\mathcal{J}$ instances with maximum~$\goal$ successful assets to~$\mathcal{P}$ instances with maximum~$\goal$ true clauses and vice versa), this reduction is an \LRed and this completes the proof.
\end{proof}

\end{document}